\documentclass[]{eptcs}
\pdfoutput=1

\usepackage{amssymb,nicefrac,amsthm}
\usepackage{amsfonts}
\usepackage{graphicx}
\usepackage[fleqn]{amsmath}
\usepackage[varg]{txfonts}
\usepackage{stmaryrd}

\usepackage{framed}
\usepackage{color}
\usepackage[normalem]{ulem}
\usepackage{tikzfig}
\tikzstyle{env}=[copoint,regular polygon rotate=0,minimum width=0.2cm, fill=black]

\tikzstyle{every picture}=[baseline=-0.25em]
\tikzstyle{dotpic}=[scale=0.5]
\tikzstyle{diredges}=[every to/.style={diredge}]
\tikzstyle{dot graph}=[shorten <=-0.1mm,shorten >=-0.1mm,scale=0.6]
\tikzstyle{plot point}=[circle,fill=black,minimum width=2mm,inner sep=0]

\tikzstyle{braceedge}=[decorate,decoration={brace,amplitude=2mm,raise=-1mm}]
\tikzstyle{small braceedge}=[decorate,decoration={brace,amplitude=1mm,raise=-1mm}]
\tikzstyle{left hook arrow}=[left hook-latex]
\tikzstyle{right hook arrow}=[right hook-latex]

\tikzstyle{black dot}=[inner sep=0.7mm,minimum width=0pt,minimum height=0pt,fill=black,draw=black,shape=circle]

\tikzstyle{dot}=[black dot]
\tikzstyle{smalldot}=[inner sep=0.4mm,minimum width=0pt,minimum height=0pt,fill=black,draw=black,shape=circle]%NEW
\tikzstyle{white dot}=[dot,fill=white]
\tikzstyle{antipode}=[white dot,inner sep=0.3mm,font=\footnotesize]
\tikzstyle{smallwhitedot}=[smalldot,fill=white]%NEW
\tikzstyle{alt white dot}=[white dot,label={[xshift=3.07mm,yshift=-0.05mm,font=\footnotesize]left:$*$}]
\tikzstyle{gray dot}=[dot,fill=gray!40!white]
\tikzstyle{smallgraydot}=[smalldot,fill=gray!40!white]%NEW
\tikzstyle{box vertex}=[draw=black,rectangle]
\tikzstyle{small box}=[box vertex,fill=white]%% added rwd]
\tikzstyle{whitebg}=[fill=white,inner sep=2pt]
\tikzstyle{graph state vertex}=[sg vertex,fill=black]

\tikzstyle{wide copoint}=[fill=white,draw=black,shape=isosceles triangle,shape border rotate=90,isosceles triangle stretches=true,inner sep=1pt,minimum width=1.5cm,minimum height=5mm]
\tikzstyle{wide point}=[fill=white,draw=black,shape=isosceles triangle,shape border rotate=-90,isosceles triangle stretches=true,inner sep=1pt,minimum width=1.5cm,minimum height=4mm]
\tikzstyle{very wide copoint}=[fill=white,draw=black,shape=isosceles triangle,shape border rotate=-90,isosceles triangle stretches=true,inner sep=1pt,minimum width=2.5cm,minimum height=4mm]
\tikzstyle{very wide empty copoint}=[draw=black,shape=isosceles triangle,shape border rotate=-90,isosceles triangle stretches=true,inner sep=1pt,minimum width=2.5cm,minimum height=4mm]
\tikzstyle{symm}=[ultra thick,shorten <=-1mm,shorten >=-1mm]

\tikzstyle{square box}=[rectangle,fill=white,draw=black,minimum height=5mm,minimum width=5mm,font=\small]
\tikzstyle{square gray box}=[rectangle,fill=gray!30,draw=black,minimum height=6mm,minimum width=6mm]
\tikzstyle{copoint}=[regular polygon,regular polygon sides=3,draw=black,scale=0.75,inner sep=-0.5pt,minimum width=7mm,fill=white]
\tikzstyle{point}=[regular polygon,regular polygon sides=3,draw=black,scale=0.75,inner sep=-0.5pt,minimum width=7mm,fill=white,regular polygon rotate=180]
\tikzstyle{gray point}=[point,fill=gray!40!white]
\tikzstyle{gray copoint}=[copoint,fill=gray!40!white]

\newcommand{\edgearrow}{{\arrow[black]{>}}}
\newcommand{\edgetick}{{\arrow[black,scale=0.7,very thick]{|}}}
\tikzstyle{diredge}=[->]
\tikzstyle{rdiredge}=[<-]
\tikzstyle{medium diredge}=[->]

\tikzstyle{short diredge}=[->]
\tikzstyle{halfedge}=[-)]
\tikzstyle{other halfedge}=[(-]
\tikzstyle{freeedge}=[(-)]
\tikzstyle{white edge}=[line width=5pt,white]
\tikzstyle{tick}=[postaction=decorate,decoration={markings, mark=at position 0.5 with \edgetick}]
\tikzstyle{small map edge}=[|-latex, gray!60!blue, shorten <=0.9mm, shorten >=0.5mm]
\tikzstyle{thick dashed edge}=[very thick,dashed,gray!40]
\tikzstyle{map edge}=[|-latex,very thick, gray!40, shorten <=1mm, shorten >=0.5mm]
\tikzstyle{tickedge}=[postaction=decorate,
  decoration={markings, mark=at position 0.5 with \edgetick}]
\tikzstyle{dirtickedge}=[postaction=decorate,
  decoration={markings, mark=at position 0.5 with \edgetick},
  decoration={markings, mark=at position 0.85 with \edgearrow}]
\tikzstyle{dirdoubletickedge}=[postaction=decorate,
  decoration={markings, mark=at position 0.4 with \edgetick},
  decoration={markings, mark=at position 0.6 with \edgetick},
  decoration={markings, mark=at position 0.85 with \edgearrow}]

\makeatletter
\newcommand{\boxshape}[3]{%
\pgfdeclareshape{#1}{
\inheritsavedanchors[from=rectangle] % this is nearly a rectangle
\inheritanchorborder[from=rectangle]
\inheritanchor[from=rectangle]{center}
\inheritanchor[from=rectangle]{north}
\inheritanchor[from=rectangle]{south}
\inheritanchor[from=rectangle]{west}
\inheritanchor[from=rectangle]{east}
\backgroundpath{% this is new
\southwest \pgf@xa=\pgf@x \pgf@ya=\pgf@y
\northeast \pgf@xb=\pgf@x \pgf@yb=\pgf@y

\@tempdima=#2
\@tempdimb=#3

\pgfpathmoveto{\pgfpoint{\pgf@xa - 5pt + \@tempdima}{\pgf@ya}}
\pgfpathlineto{\pgfpoint{\pgf@xa - 5pt - \@tempdima}{\pgf@yb}}
\pgfpathlineto{\pgfpoint{\pgf@xb + 5pt + \@tempdimb}{\pgf@yb}}
\pgfpathlineto{\pgfpoint{\pgf@xb + 5pt - \@tempdimb}{\pgf@ya}}
\pgfpathlineto{\pgfpoint{\pgf@xa - 5pt + \@tempdima}{\pgf@ya}}
\pgfpathclose
}
}}

\boxshape{NEbox}{0pt}{8pt}
\boxshape{SEbox}{0pt}{-8pt}
\boxshape{NWbox}{8pt}{0pt}
\boxshape{SWbox}{-8pt}{0pt}
\makeatother

\tikzstyle{map}=[draw,shape=NEbox,inner sep=7pt]
\tikzstyle{mapdag}=[draw,shape=SEbox,inner sep=7pt]
\tikzstyle{maptrans}=[draw,shape=SWbox,inner sep=7pt]
\tikzstyle{mapconj}=[draw,shape=NWbox,inner sep=7pt]

\tikzstyle{probs}=[shape=semicircle,fill=gray!40!white,draw=black,shape border rotate=180,minimum width=1.2cm]

\tikzstyle{arrs}=[-latex,font=\small,auto]
\tikzstyle{arrow plain}=[arrs]
\tikzstyle{arrow dashed}=[dashed,arrs]
\tikzstyle{arrow bold}=[very thick,arrs]
\tikzstyle{arrow hide}=[draw=white!0,-]
\tikzstyle{arrow reverse}=[latex-]
\tikzstyle{cdnode}=[]

\tikzstyle{gn}=[dot,fill=lime!50,minimum width=0.2cm,inner sep=0.5pt,font=\footnotesize]
\tikzstyle{rn}=[dot,fill=red!50,inner sep=0.5pt,minimum width=0.2cm,font=\footnotesize]
\tikzstyle{bn}=[dot,fill=blue,minimum width=0.3cm]

\tikzstyle{rc}=[dot,thick,fill=white,draw = red,minimum width=0.2cm,inner sep=0.5pt,font=\footnotesize]
\tikzstyle{gc}=[dot,thick,fill=white,draw= lime,inner sep=0.5pt,minimum width=0.2cm,font=\footnotesize]
\tikzstyle{bc}=[dot,thick,fill=white,draw= blue,minimum width=0.3cm]

\tikzstyle{label}=[circle,fill=white,minimum width=0.3cm]

\tikzstyle{H box}=[rectangle,fill=yellow,draw=black,xscale=1,yscale=1,font=\small,inner sep=0.75pt,minimum width=0.15cm,minimum height=0.15cm]

\tikzstyle{clocklabel}=[dot,fill=yellow,draw=black,font=\tiny,inner sep=0.75pt]

\tikzstyle{rsn}=[circle split,draw,fill=red,font=\tiny,inner sep=0.75pt]
\tikzstyle{gsn}=[circle split,draw,fill=lime,font=\tiny,inner sep=0.75pt]
\tikzstyle{bsn}=[circle split,draw,fill=blue,font=\tiny,inner sep=0.75pt]

\tikzstyle{rsc}=[circle split,thick,draw= red,draw,fill=white,font=\tiny,inner sep=0.75pt]
\tikzstyle{gsc}=[circle split,thick,draw= lime,draw,fill=white,font=\tiny,inner sep=0.75pt]
\tikzstyle{bsc}=[circle split,thick,draw= blue,draw,fill=white,font=\tiny,inner sep=0.75pt]

\tikzstyle{cnot}=[fill=white,shape=circle,inner sep=-1.4pt]
\tikzstyle{wire label}=[font=\tiny, auto]

\newcommand{\denoteb}[1]{% --``semantic'' brakets
\left\llbracket #1 \right\rrbracket} 

\newcommand{\den}[1]{\denoteb{#1}^\natural}

\newcommand{\denote}[1]{% --``semantic'' brakets
\llbracket #1 \rrbracket} 

\newcommand{\dempty}{%
\beginpgfgraphicnamed{scalars-s/emptysquare-small}
\InputIfFileExists{scalars-s/emptysquare-small.tikz}{}{\input{./figures/scalars-s/emptysquare-small.tikz}}
\endpgfgraphicnamed}

\newcommand{\drcup}{\raisebox{-0.15cm}{%
\beginpgfgraphicnamed{scalars-s/cup}
\InputIfFileExists{scalars-s/cup.tikz}{}{\input{./figures/scalars-s/cup.tikz}}
\endpgfgraphicnamed}}
\newcommand{\drcap}{\raisebox{0.15cm}{%
\beginpgfgraphicnamed{scalars-s/cap}
\InputIfFileExists{scalars-s/cap.tikz}{}{\input{./figures/scalars-s/cap.tikz}}
\endpgfgraphicnamed}}

\newcommand{\dsigmaf}{%
\beginpgfgraphicnamed{scalars-s/swapf}
\InputIfFileExists{scalars-s/swapf.tikz}{}{\input{./figures/scalars-s/swapf.tikz}}
\endpgfgraphicnamed}
\newcommand{\dHf}{%
\beginpgfgraphicnamed{scalars-s/Had4f}
\begin{tikzpicture}
	\begin{pgfonlayer}{nodelayer}
		\node [style={H box}] (0) at (0.5, 0) {};
		\node [style=none] (1) at (0.5, 0.3) {};
		\node [style=none] (2) at (0.5, -0.3) {};
	\end{pgfonlayer}
	\begin{pgfonlayer}{edgelayer}
		\draw (1.center) to (0);
		\draw (2.center) to (0);
	\end{pgfonlayer}
\end{tikzpicture}}
\endpgfgraphicnamed}
\newcommand{\drcupf}{\raisebox{-0.15cm}{%
\beginpgfgraphicnamed{scalars-s/cupf}
\InputIfFileExists{scalars-s/cupf.tikz}{}{\input{./figures/scalars-s/cupf.tikz}}
\endpgfgraphicnamed}}
\newcommand{\drcapf}{\raisebox{0.15cm}{%
\beginpgfgraphicnamed{scalars-s/capf}
\InputIfFileExists{scalars-s/capf.tikz}{}{\input{./figures/scalars-s/capf.tikz}}
\endpgfgraphicnamed}}

\newcommand{\did}{%
\beginpgfgraphicnamed{scalars-s/Id}
\begin{tikzpicture}
	\begin{pgfonlayer}{nodelayer}
		\node [style=none] (1) at (0.5, 0.3) {};
		\node [style=none] (2) at (0.5, -0.3) {};
		\node [style=none] (3) at (0.5, -0.5) {};
		\node [style=none] (4) at (0.5, 0.5) {};
	\end{pgfonlayer}
	\begin{pgfonlayer}{edgelayer}
		\draw (1.center) to (2.center);
	\end{pgfonlayer}
\end{tikzpicture}}
\endpgfgraphicnamed}
\newcommand{\didf}{%
\beginpgfgraphicnamed{scalars-s/Idf}
\begin{tikzpicture}
	\begin{pgfonlayer}{nodelayer}
		\node [style=none] (1) at (0.5, 0.3) {};
		\node [style=none] (2) at (0.5, -0.3) {};
	\end{pgfonlayer}
	\begin{pgfonlayer}{edgelayer}
		\draw (1.center) to (2.center);
	\end{pgfonlayer}
\end{tikzpicture}}
\endpgfgraphicnamed}

\newcommand{\gdzz}{%
\beginpgfgraphicnamed{scalars-s/RZ00alpha}
\begin{tikzpicture}
	\begin{pgfonlayer}{nodelayer}
		\node [style=gn] (0) at (0.5, 0) {\footnotesize$\alpha$};
	\end{pgfonlayer}
\end{tikzpicture}}
\endpgfgraphicnamed}
\newcommand{\gdzo}{%
\beginpgfgraphicnamed{scalars-s/RZ01alpha}
\begin{tikzpicture}
	\begin{pgfonlayer}{nodelayer}
		\node [style=gn] (0) at (0.5, 0.1) {\footnotesize$\alpha$};
		\node [style=none] (1) at (0.5, -0.2) {};
	\end{pgfonlayer}
	\begin{pgfonlayer}{edgelayer}
		\draw (1.center) to (0.center);
	\end{pgfonlayer}

\end{tikzpicture}}
\endpgfgraphicnamed}

\newcommand{\gdoo}{%
\beginpgfgraphicnamed{scalars-s/RZ11alpha}
\begin{tikzpicture}
	\begin{pgfonlayer}{nodelayer}
		\node [style=gn] (0) at (0.5, 0) {\footnotesize$\alpha$};
		\node [style=none] (1) at (0.5, -0.3) {};
		\node [style=none] (2) at (0.5, 0.3) {};
	\end{pgfonlayer}
	\begin{pgfonlayer}{edgelayer}
		\draw (1.center) to (0.center);
				\draw (2.center) to (0.center);
	\end{pgfonlayer}

\end{tikzpicture}}
\endpgfgraphicnamed}
\newcommand{\gdto}{%
\beginpgfgraphicnamed{scalars-s/RZ21alpha}
\InputIfFileExists{scalars-s/RZ21alpha.tikz}{}{\input{./figures/scalars-s/RZ21alpha.tikz}}
\endpgfgraphicnamed}
\newcommand{\gpi}{%
\beginpgfgraphicnamed{scalars-s/RZ00pi}
\begin{tikzpicture}
	\begin{pgfonlayer}{nodelayer}
		\node [style=gn] (0) at (0.5, 0) {\footnotesize$\pi$};
	\end{pgfonlayer}
\end{tikzpicture}}
\endpgfgraphicnamed}

\newcommand{\rdzz}{%
\beginpgfgraphicnamed{scalars-s/RX00alpha}
\begin{tikzpicture}
	\begin{pgfonlayer}{nodelayer}
		\node [style=rn] (0) at (0.5, 0) {\footnotesize$\alpha$};
	\end{pgfonlayer}
\end{tikzpicture}}
\endpgfgraphicnamed}
\newcommand{\rdzo}{%
\beginpgfgraphicnamed{scalars-s/RX01alpha}
\begin{tikzpicture}
	\begin{pgfonlayer}{nodelayer}
		\node [style=rn] (0) at (0.5, 0.1) {\footnotesize$\alpha$};
		\node [style=none] (1) at (0.5, -0.2) {};
	\end{pgfonlayer}
	\begin{pgfonlayer}{edgelayer}
		\draw (1.center) to (0.center);
	\end{pgfonlayer}

\end{tikzpicture}}
\endpgfgraphicnamed}
\newcommand{\rdoo}{%
\beginpgfgraphicnamed{scalars-s/RX11alpha}
\begin{tikzpicture}
	\begin{pgfonlayer}{nodelayer}
		\node [style=rn] (0) at (0.5, 0) {\footnotesize$\alpha$};
		\node [style=none] (1) at (0.5, -0.3) {};
		\node [style=none] (2) at (0.5, 0.3) {};
	\end{pgfonlayer}
	\begin{pgfonlayer}{edgelayer}
		\draw (1.center) to (0.center);
				\draw (2.center) to (0.center);
	\end{pgfonlayer}

\end{tikzpicture}}
\endpgfgraphicnamed}

\tikzstyle{cdiag}=[matrix of math nodes, row sep=3em, column sep=3em, text height=1.5ex, text depth=0.25ex,inner sep=0.5em]
\tikzstyle{arrow above}=[transform canvas={yshift=0.5ex}]
\tikzstyle{arrow below}=[transform canvas={yshift=-0.5ex}]

\usepackage{mathtools}
\newtheorem{Th}{Theorem}[section]
\newtheorem{theorem}[Th]{Theorem}

\newtheorem{lemma}[Th]{Lemma}
\newtheorem{corollary}[Th]{Corollary}

\newtheorem{remark}[Th]{Remark}
\tikzstyle{square green box}=[rectangle,fill=green!30,draw=black]%
\tikzstyle{square red box}=[rectangle,fill=red!30,draw=black]%
\newcommand{\ZX}{\textsc{ZX}}
\tikzstyle{bscalar}=[star,fill=black,draw=black,minimum size=8pt,inner sep=0pt]
\title{A Simplified Stabilizer \ZX-calculus}

\author{Miriam Backens
\institute{School of Mathematics, University of Bristol,  UK}
\email{m.backens@bristol.ac.uk}
\and Simon Perdrix
\institute{CNRS, LORIA, Universit\'e de Lorraine, France}
\email{simon.perdrix@loria.fr}
\and Quanlong Wang
\institute{LORIA, Universit\'e de Lorraine, France\\ Department of Computer Science, University of Oxford, UK }
\email{quanlong.wang@wolfson.ox.ac.uk}
}

\def\titlerunning{A Simplified Stabilizer \ZX-calculus}
\def\authorrunning{Miriam Backens, Simon Perdrix \& Quanlong Wang}

\begin{document}

\maketitle

\begin{abstract}
The stabilizer \ZX-calculus is a rigorous graphical language for reasoning about quantum mechanics. The language is sound and complete: a stabilizer \ZX-diagram can be transformed into another one if and only if these two diagrams represent the same quantum evolution or quantum state. We show that the stabilizer \ZX-calculus can be simplified, removing unnecessary equations while keeping only the essential axioms which potentially capture  fundamental structures of quantum mechanics. We thus give a significantly smaller set of axioms and prove that meta-rules like `colour symmetry' and `upside-down symmetry', which were considered as axioms in previous versions of the language, can in fact be derived. In particular, we show that the additional symbol and one of the rules which had been recently introduced to keep track of scalars (diagrams with no inputs or outputs) are not necessary.
\end{abstract}

\section{Introduction}
The \ZX-calculus is a high-level and intuitive graphical language for pure qubit quantum mechanics (QM), based on category theory \cite{coecke_interacting_2011}. It comes with a set of rewrite rules that potentially allow this graphical calculus to replace matrix-based formalisms entirely for certain classes of problems. However, this replacement is only possible without losing deductive power if the \ZX-calculus is \emph{complete} for this class of problems, i.e.\ if any equality that is derivable using matrices can also be derived graphically.

The overall \ZX-calculus for pure state qubit quantum mechanics is incomplete, and it is not obvious
how to complete it \cite{schroeder_incomplete_2014}. Yet, a fragment of the language, the \emph{stabilizer \ZX-calculus} is complete. This fragment is made of the \ZX-diagrams involving angles which are multiples of $\pi/2$ only. The fragment of quantum theory that can be represented by stabilizer \ZX-diagrams is the so-called stabilizer quantum mechanics \cite{gottesman_stabilizer_1997}. Stabilizer QM is a non trivial fragment of quantum mechanics which is in fact efficiently classically simulatable \cite{gottesman_heisenberg_1998}   but which nevertheless exhibits
many important quantum properties, like entanglement and non-locality. It is furthermore of central importance in areas such as quantum error correcting codes \cite{nielsen_quantum_2010} and measurement-based quantum computation \cite{raussendorf_one-way_2001}.

The stabilizer \ZX-calculus is the largest known complete fragment of the \ZX-calculus. Furthermore, it is the core of the overall language since all the fundamental structures -- e.g.\ the axiomatisation of complementary bases \cite{coecke_interacting_2011} -- are present in this fragment. 

Here, we simplify the stabilizer \ZX-calculus, removing unnecessary equations while keeping only the essential axioms, which potentially capture  fundamental structures of quantum mechanics. Simplifying the \ZX-calculus also simplifies the development, and potentially the  efficiency, of automated tools for quantum reasoning, e.g.\ Quantomatic \cite{quanto}. 
We give a set of axioms that is significantly smaller than the usual one and prove that meta-rules like `colour symmetry' and `upside-down symmetry', which were considered as axioms in previous versions of the language can in fact be derived.

This is done by first considering the scalars (Section \ref{s:minimal_scalars}), i.e.\ the diagrams with no inputs or outputs. The completeness of the stabilizer \ZX-calculus was originally established in a setting where scalars are ignored \cite{backens_zx-calculus_2013}. In this setting, when two diagrams are equal according to the rules of the language, their matrices are equal up to a non-zero scalar factor. To make the stabilizer \ZX-calculus with scalars complete, a new symbol  and three rules were added to the original \ZX-calculus \cite{backens_making_2015}. We simplify the completeness proof for the  stabilizer \ZX-calculus by showing that there is no need to introduce a new symbol to make the language complete for scalars. Moreover, we show that one of the three new axioms is not necessary and can be derived from the rest of the language. We end up with only two rules explicitly about scalars, which we prove to be necessary. 

Beyond the treatment of scalars, we also consider a simplified set of rules for the full stabilizer \ZX-calculus (Section \ref{s:simplified}). Usually, in addition to about a dozen explicit rewrite rules, there is a convention that any rule also holds with the colours red and green swapped or with the diagrams flipped upside-down, effectively nearly quadrupling the available set of rewrite rules\footnote{Some rules are symmetric under the operations of swapping the colours and/or flipping them upside-down, hence the effective rule set is not quite four times the size of the explicitly-given one.}. Here, we give a new system of just nine rules for the stabilizer \ZX-calculus with scalars. We prove that these rules are sound and that all the old rules, including their colour-swapped and upside-down versions, can be derived from the new system of rules.

\section{Stabilizer \ZX-calculus}\label{sec:zx}

\subsection{Diagrams and standard interpretation}

A diagram $D:k\to l$ of the stabilizer \ZX-calculus with $k$ inputs and $l$ outputs is generated by:
\begin{center}
\begin{tabular}{|r@{~}r@{~}c@{~}lc|r@{~}r@{~}c@{~}lc|r@{~}r@{~}c@{~}lc|}
\hline
&$R_Z^{(n,m)}(\alpha)$&$:$&$n\to m$ & %
\beginpgfgraphicnamed{scalars-s/spideralpha}
\InputIfFileExists{scalars-s/spideralpha.tikz}{}{\input{./figures/scalars-s/spideralpha.tikz}}
\endpgfgraphicnamed & &$R_X^{(n,m)}(\alpha)$&$:$&$ n\to m$& %
\beginpgfgraphicnamed{scalars-s/spiderredalpha}
\InputIfFileExists{scalars-s/spiderredalpha.tikz}{}{\input{./figures/scalars-s/spiderredalpha.tikz}}
\endpgfgraphicnamed& &$H$&$:$&$1\to 1$ &%
\beginpgfgraphicnamed{scalars-s/Had4}
\InputIfFileExists{scalars-s/Had4.tikz}{}{\input{./figures/scalars-s/Had4.tikz}}
\endpgfgraphicnamed\\
\hline
& $s$&$:$&$0\to 0$ &%
\beginpgfgraphicnamed{scalars/halfscalar}
\begin{tikzpicture}
	\begin{pgfonlayer}{nodelayer}
		\node [style=bscalar] (1) at (0, 0) {};
	\end{pgfonlayer}
\end{tikzpicture}
}
\endpgfgraphicnamed  &
 & $e $&$:$&$0 \to 0$ &%
\beginpgfgraphicnamed{scalars-s/emptysquare-small}
\InputIfFileExists{scalars-s/emptysquare-small.tikz}{}{\input{./figures/scalars-s/emptysquare-small.tikz}}
\endpgfgraphicnamed & &$\sigma$&$:$&$ 2\to 2$& %
\beginpgfgraphicnamed{scalars-s/swap}
\InputIfFileExists{scalars-s/swap.tikz}{}{\input{./figures/scalars-s/swap.tikz}}
\endpgfgraphicnamed \\\hline
 &$\mathbb I$&$:$&$1\to 1$&%
\beginpgfgraphicnamed{scalars-s/Id}
\begin{tikzpicture}
	\begin{pgfonlayer}{nodelayer}
		\node [style=none] (1) at (0.5, 0.3) {};
		\node [style=none] (2) at (0.5, -0.3) {};
		\node [style=none] (3) at (0.5, -0.5) {};
		\node [style=none] (4) at (0.5, 0.5) {};
	\end{pgfonlayer}
	\begin{pgfonlayer}{edgelayer}
		\draw (1.center) to (2.center);
	\end{pgfonlayer}
\end{tikzpicture}}
\endpgfgraphicnamed &
  &$\epsilon$&$:$&$2\to 0$& %
\beginpgfgraphicnamed{scalars-s/cup}
\InputIfFileExists{scalars-s/cup.tikz}{}{\input{./figures/scalars-s/cup.tikz}}
\endpgfgraphicnamed&
  &$\eta$&$:$&$ 0\to 2$&  %
\beginpgfgraphicnamed{scalars-s/cap}
\InputIfFileExists{scalars-s/cap.tikz}{}{\input{./figures/scalars-s/cap.tikz}}
\endpgfgraphicnamed
  \\\hline
\end{tabular}\\where $m,n\in \mathbb N$ and $\alpha \in \{0,\frac \pi 2, \pi, \frac {-\pi} 2\}$
\end{center}

\begin{itemize}
\item Spacial composition: for any $D_1:a\to b$ and $D_2: c\to d$, $D_1\otimes D_2 : a+c\to b+d$ is constructed by placing $D_1$ and $D_2$ side-by-side, $D_2$ to the right of $D_1$.
\item Sequential composition: for any $D_1:a\to b$ and $D_2: b\to c$, $D_2\circ D_1 : a\to c$ is constructed by placing $D_1$ above $D_2$, connecting the outputs of $D_1$ to the inputs of $D_2$.
\end{itemize}

When equal to $0$, the phase angles of the green and red dots may be omitted:
$$%
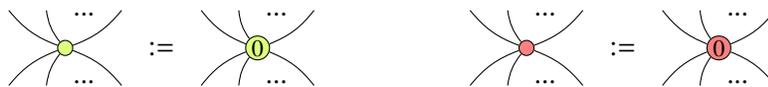
\beginpgfgraphicnamed{scalars-s/spiderg}
\InputIfFileExists{scalars-s/spiderg.tikz}{}{\input{./figures/scalars-s/spiderg.tikz}}
\endpgfgraphicnamed := %
\beginpgfgraphicnamed{scalars-s/spidergz}
\InputIfFileExists{scalars-s/spidergz.tikz}{}{\input{./figures/scalars-s/spidergz.tikz}}
\endpgfgraphicnamed\qquad\qquad%
\beginpgfgraphicnamed{scalars-s/spiderr}
\InputIfFileExists{scalars-s/spiderr.tikz}{}{\input{./figures/scalars-s/spiderr.tikz}}
\endpgfgraphicnamed := %
\beginpgfgraphicnamed{scalars-s/spiderrz}
\InputIfFileExists{scalars-s/spiderrz.tikz}{}{\input{./figures/scalars-s/spiderrz.tikz}}
\endpgfgraphicnamed\qquad$$

The standard interpretation of the \ZX-diagrams associates with any diagram $D:n\to m$ a linear map $\denoteb{D}:\mathbb C^{2^n}\to \mathbb C^{2^m}$ inductively defined as follows:
\begin{align*}
  \denoteb{D_1\otimes D_2}&:=\denoteb{D_1}\otimes \denoteb{D_2} &
  \denoteb{~\dHf~}&:=\frac{1}{\sqrt{2}}\left(\begin{array}{@{}c@{}r@{}}1&1\\1&{~\text{-}}1\end{array}\right) &
  \denoteb{~\dempty~}&:=1 \\
  \denoteb{D_2\circ D_1}&:=\denoteb{D_2}\circ\denoteb{D_1} &
  \denoteb{~%
\beginpgfgraphicnamed{scalars/halfscalar}
\begin{tikzpicture}
	\begin{pgfonlayer}{nodelayer}
		\node [style=bscalar] (1) at (0, 0) {};
	\end{pgfonlayer}
\end{tikzpicture}
}
\endpgfgraphicnamed~}&:=\frac{1}{2} &
  \denote{\drcupf}&:=\left(\begin{array}{@{}c@{}c@{}@{}c@{}c@{}}1&0&0&1\end{array}\right) \\
  \denoteb{~\dsigmaf~}&:=\left(\begin{array}{@{}c@{}r@{}@{}c@{}c@{}}1&0&0&0\\0&0&1&0\\0&1&0&0\\0&0&0&1\\\end{array}\right) &
  \denoteb{~\did~}&:=\left(\begin{array}{@{}c@{}c@{}}1&0\\0&1\end{array}\right) &
  \denoteb{\drcapf}&:=\left(\begin{array}{@{}c@{}}1\\0\\0\\1\end{array}\right)
\end{align*}
For green dots, $\denote{R_Z^{(0,0)}(\alpha)}:= 1{+}e^{i\alpha}$, and when $a{+}b>0$, $\denote{R_Z^{(a,b)}(\alpha)}$ is a matrix with $2^a$ columns and $2^b$ rows such that all entries are $0$ except the top left one which is $1$ and the bottom right one which is $e^{i\alpha}$, e.g.:
$$
\def\arraystretch{0.5}
\denoteb{~\gdzz~} = 1+e^{i\alpha} \qquad \denoteb{~\gdzo~} = \left(\begin{array}{@{}c@{}}1\\e^{i\alpha}\end{array}\right) \qquad \denoteb{~\gdoo~} = \left(\begin{array}{@{}c@{}c@{}}1&~0~\\0&~e^{i\alpha}\end{array}\right) \qquad  \denoteb{\gdto} = \left(\begin{array}{@{}c@{}c@{}c@{}c@{}}1~&~0~&~0~&~0~\\0~&~0~&~0~&~e^{i\alpha}\end{array}\right) $$

\noindent For any $a,b\ge  0$, $\denote{R_X^{a,b}(\alpha)}:= \denote{H}^{\otimes b}\times \denote{R_Z^{a,b}(\alpha)} \times \denote{H}^{\otimes a}$, where $M^{\otimes 0} =1$ and for any $k>0$, $M^{\otimes k}=M\otimes M^{\otimes k-1}$. E.g.,
$$
\def\arraystretch{0.5}
\denoteb{~\rdzz~} = 1+e^{i\alpha} \qquad \denoteb{~\rdzo~}= \sqrt2 e^{i\frac \alpha 2}\!\left(\begin{array}{@{}r@{}}\cos(\nicefrac{\alpha}2)\\\text{-}i\sin(\nicefrac{\alpha}2)\end{array}\right) \qquad \denoteb{~\rdoo~} = e^{i\frac \alpha 2}\! \left(\begin{array}{@{}r@{}r@{}}\cos(\nicefrac{\alpha}2)&\text{~~-}i\sin(\nicefrac{\alpha}2)\\\text{-}i\sin(\nicefrac{\alpha}2)&\cos(\nicefrac{\alpha}2)\end{array}\right)
$$

The linear maps that can be represented by stabilizer \ZX-diagrams correspond to the so-called stabilizer fragment of quantum mechanics \cite{gottesman_stabilizer_1997}. Note that \ZX-diagrams with arbitrary angles (no longer necessarily multiples of $\frac \pi 2$) are universal:  for any $m,n\ge 0$ and any linear map $M:\mathbb C^{2^n}\to \mathbb C^{2^m}$, there exists a diagram $D:n\to m$ such that $\denote{D} = M$ \cite{coecke_interacting_2011}. When restricted to angles that are multiples of $\pi/4$, \ZX-diagrams are approximately universal, i.e. any linear map can approximated to arbitrary accuracy by such a \ZX-diagram. In this paper, we focus on the core of the \ZX-calculus formed by the stabilizer \ZX-diagrams.

\subsection{The rewrite rules}

\begin{figure}[ht]
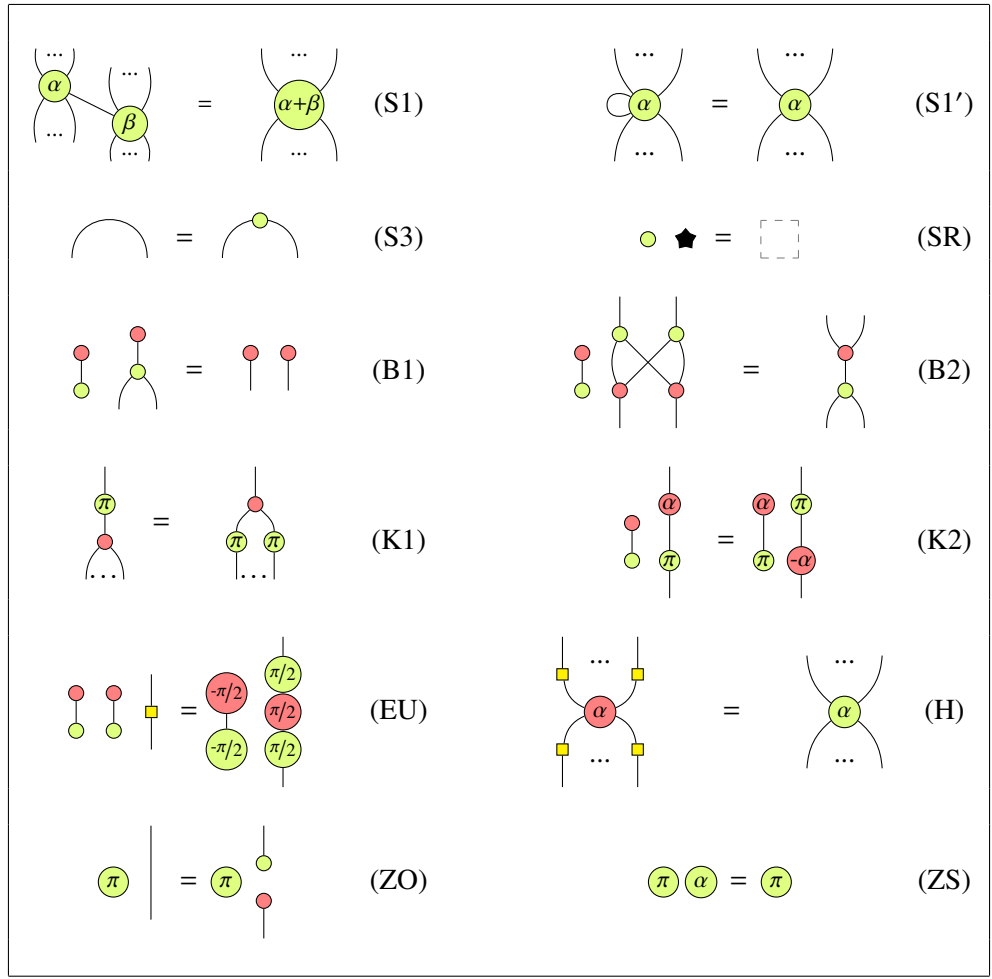

 \centering
  \begin{tabular}{|ccccc|}
   \hline
   &&&& \\
\beginpgfgraphicnamed{scalars/spider-bis}
\InputIfFileExists{scalars/spider-bis.tikz}{}{\input{./figures/scalars/spider-bis.tikz}}
\endpgfgraphicnamed&(S1) &$\qquad$& %
\beginpgfgraphicnamed{scalars/spider2-loop}
\InputIfFileExists{scalars/spider2-loop.tikz}{}{\input{./figures/scalars/spider2-loop.tikz}}
\endpgfgraphicnamed&(S1$'$)\\
   &&&& \\
\beginpgfgraphicnamed{scalars/induced_compact_structure-2wire}
\InputIfFileExists{scalars/induced_compact_structure-2wire.tikz}{}{\input{./figures/scalars/induced_compact_structure-2wire.tikz}}
\endpgfgraphicnamed&(S3) && %
\beginpgfgraphicnamed{scalars/star_rule}
\InputIfFileExists{scalars/star_rule.tikz}{}{\input{./figures/scalars/star_rule.tikz}}
\endpgfgraphicnamed&(SR)\\
   &&&& \\
\beginpgfgraphicnamed{scalars/b1s}
\InputIfFileExists{scalars/b1s.tikz}{}{\input{./figures/scalars/b1s.tikz}}
\endpgfgraphicnamed&(B1) && %
\beginpgfgraphicnamed{scalars/b2s}
\InputIfFileExists{scalars/b2s.tikz}{}{\input{./figures/scalars/b2s.tikz}}
\endpgfgraphicnamed&(B2)\\
   &&&& \\
\beginpgfgraphicnamed{scalars/k1}
\InputIfFileExists{scalars/k1.tikz}{}{\input{./figures/scalars/k1.tikz}}
\endpgfgraphicnamed&(K1) && %
\beginpgfgraphicnamed{scalars/k2s}
\InputIfFileExists{scalars/k2s.tikz}{}{\input{./figures/scalars/k2s.tikz}}
\endpgfgraphicnamed&(K2)\\
   &&&& \\
\beginpgfgraphicnamed{scalars/HadaDecomSingles}
\InputIfFileExists{scalars/HadaDecomSingles.tikz}{}{\input{./figures/scalars/HadaDecomSingles.tikz}}
\endpgfgraphicnamed&(EU) && %
\beginpgfgraphicnamed{scalars/h2}
\InputIfFileExists{scalars/h2.tikz}{}{\input{./figures/scalars/h2.tikz}}
\endpgfgraphicnamed&(H)\\
   &&&& \\
\beginpgfgraphicnamed{scalars/zo1}
\InputIfFileExists{scalars/zo1.tikz}{}{\input{./figures/scalars/zo1.tikz}}
\endpgfgraphicnamed&(ZO) && %
\beginpgfgraphicnamed{scalars/zero_scalar}
\InputIfFileExists{scalars/zero_scalar.tikz}{}{\input{./figures/scalars/zero_scalar.tikz}}
\endpgfgraphicnamed&(ZS)\\
   &&&& \\
   \hline
  \end{tabular}
 \caption{Rules for the stabilizer \ZX-calculus with scalars. All of these rules also hold when flipped upside-down, or with the colours red and green swapped. The right-hand side of (SR) is an empty diagram. Ellipses denote zero or more wires. The sum in (S1) is modulo $2\pi$. }
 \label{fig:ZX_rules}
\end{figure}

The \ZX-calculus is not just a notation: it comes with a set of rewrite rules that allow equalities to be derived entirely graphically.
The reason we are considering the stabilizer \ZX-calculus here is that, for this theory, a \emph{complete} set of rewrite rules is known: this means that any equality that can be derived using matrices can also be derived graphically using that set of rewrite rules \cite{backens_zx-calculus_2013,backens_making_2015}.
On the other hand, the currently-used set of rewrite rules is known to be incomplete for the universal \ZX-calculus, but it is unclear how to complete it \cite{schroeder_incomplete_2014,PW15}.

The set of rewrite rules for the stabilizer \ZX-calculus with scalars -- as used in \cite{backens_making_2015} -- is given in Figure \ref{fig:ZX_rules}. Notice that these rules are symmetric, i.e.\ if $D_1=D_2$ then $D_2=D_1$. 
All of those rules also hold upside-down and/or with the colours red and green swapped.
Whenever a rule contains an ellipsis to indicate that it applies to spiders with different numbers of inputs or outputs, those numbers can take any non-negative integer value, including zero.
We simply denote by  $D_1=D_2$ the existence of a (possibly empty) sequence of rules which transforms $D_1$ into $D_2$.

Rules (SR), (ZO), and (ZS) were newly introduced in \cite{backens_making_2015} to deal with scalars.
The Euler decomposition rule (EU) was introduced in \cite{duncan_graph_2009}.
All other rules were part of the original definition of the \ZX-calculus, although some of them have been modified because equality in the original \ZX-calculus was only up to a global phase, i.e.\ two diagrams were considered equal if they represented matrices that differed by a scalar factor of $e^{i\phi}$ for some $\phi\in(-\pi,\pi]$.

The rules given in Figure \ref{fig:ZX_rules} apply to any sub-diagram. In other words,  if $D_1 = D_2$ then, for any $D$ (with the appropriate number of inputs/outputs), $D\otimes D_1 = D\otimes D_2$ ;   $D_1\otimes D = D_2\otimes D$ ; $D\circ D_1 = D\circ D_2$ ; and  $D_1\circ D = D_2\circ  D$.

In addition to those explicit  rules there is also a meta-rule:
 \emph{`only the topology matters'} \cite{coecke_interacting_2011},
which means that two diagrams represent the same matrix whenever one can be transformed into the other by moving components around without changing their connections.
E.g.
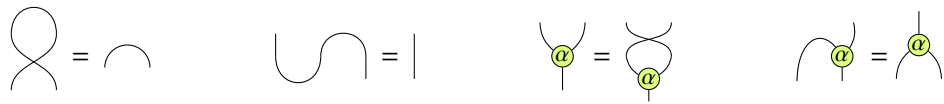
$$%
\beginpgfgraphicnamed{scalars-s/capswap}
\InputIfFileExists{scalars-s/capswap.tikz}{}{\input{./figures/scalars-s/capswap.tikz}}
\endpgfgraphicnamed ~=\drcapf  \qquad\qquad %
\beginpgfgraphicnamed{scalars-s/yanking1}
\InputIfFileExists{scalars-s/yanking1.tikz}{}{\input{./figures/scalars-s/yanking1.tikz}}
\endpgfgraphicnamed~ =~ %
\beginpgfgraphicnamed{scalars-s/yangkingline}
\begin{tikzpicture}[scale=0.6]
	\begin{pgfonlayer}{nodelayer}
		\node [style=none] (0) at (0, 0.5) {};
		\node [style=none] (1) at (0, -0.5) {};
	\end{pgfonlayer}
	\begin{pgfonlayer}{edgelayer}
		\draw (0.center) to (1.center);
	\end{pgfonlayer}
\end{tikzpicture}}
\endpgfgraphicnamed ~\qquad\qquad%
\beginpgfgraphicnamed{scalars-s/commute1}
\InputIfFileExists{scalars-s/commute1.tikz}{}{\input{./figures/scalars-s/commute1.tikz}}
\endpgfgraphicnamed =~%
\beginpgfgraphicnamed{scalars-s/commute2}
\InputIfFileExists{scalars-s/commute2.tikz}{}{\input{./figures/scalars-s/commute2.tikz}}
\endpgfgraphicnamed~\qquad\qquad %
\beginpgfgraphicnamed{scalars-s/bendingnew}
\InputIfFileExists{scalars-s/bendingnew.tikz}{}{\input{./figures/scalars-s/bendingnew.tikz}}
\endpgfgraphicnamed ~=%
\beginpgfgraphicnamed{scalars-s/nonbending}
\InputIfFileExists{scalars-s/nonbending.tikz}{}{\input{./figures/scalars-s/nonbending.tikz}}
\endpgfgraphicnamed ~$$

The  rules of the stabilizer \ZX-calculus given in Figure \ref{fig:ZX_rules} are \emph{sound}: for any diagrams $D_1$ and $D_2$, if $D_1=D_2$ according to the rules of the stabilizer \ZX-calculus, then $\denote{D_1}= \denote{D_2}$.
Furthermore, as mentioned above, this set of rewrite rules is also \emph{complete} for the stabilizer \ZX-calculus, meaning that for any two stabilizer \ZX-calculus diagrams $D_1$ and $D_2$, if $\denote{D_1} = \denote{D_2}$ then $D_1 = D_2$.

\section{Minimal axioms for scalars in the stabilizer \ZX-calculus}
\label{s:minimal_scalars}

To make the stabilizer \ZX-calculus complete for scalars, a new symbol $%
\beginpgfgraphicnamed{scalars/halfscalar}
\begin{tikzpicture}
	\begin{pgfonlayer}{nodelayer}
		\node [style=bscalar] (1) at (0, 0) {};
	\end{pgfonlayer}
\end{tikzpicture}
}
\endpgfgraphicnamed$ together with three axioms (SR), (ZS), and (ZO) were added in \cite{backens_making_2015}. We show in this section that the new symbol and one of those axioms are not actually necessary for completeness. We end up with a simplified set of rules (Figure \ref{fig:ZX_rules} minus (ZS) and with a star-free replacement (IV) for (SR)) in which only two axioms -- (IV) and (ZO) -- are dedicated to scalars. We show the minimality of those two axioms in the sense that they cannot be derived from the other rules of the language.

\subsection{Removing the star}

The symbol $%
\beginpgfgraphicnamed{scalars/halfscalar}
\begin{tikzpicture}
	\begin{pgfonlayer}{nodelayer}
		\node [style=bscalar] (1) at (0, 0) {};
	\end{pgfonlayer}
\end{tikzpicture}
}
\endpgfgraphicnamed$ was introduced in  \cite{backens_making_2015} as an inverse of $%
\beginpgfgraphicnamed{scalars/gnode}
\begin{tikzpicture}
	\begin{pgfonlayer}{nodelayer}
		\node [style=gn] (0) at (0, -0) {};
	\end{pgfonlayer}
\end{tikzpicture}}
\endpgfgraphicnamed$. In fact, there also exists a $%
\beginpgfgraphicnamed{scalars/halfscalar}
\begin{tikzpicture}
	\begin{pgfonlayer}{nodelayer}
		\node [style=bscalar] (1) at (0, 0) {};
	\end{pgfonlayer}
\end{tikzpicture}
}
\endpgfgraphicnamed$-free diagram which is equal to $%
\beginpgfgraphicnamed{scalars/halfscalar}
\begin{tikzpicture}
	\begin{pgfonlayer}{nodelayer}
		\node [style=bscalar] (1) at (0, 0) {};
	\end{pgfonlayer}
\end{tikzpicture}
}
\endpgfgraphicnamed$:

\begin{lemma}\label{starroot2}In the \ZX-calculus, i.e.\ using the rules of Figure \ref{fig:ZX_rules}:
$$%
\beginpgfgraphicnamed{scalars/3linesquare}
\InputIfFileExists{scalars/3linesquare.tikz}{}{\input{./figures/scalars/3linesquare.tikz}}
\endpgfgraphicnamed=%
\beginpgfgraphicnamed{scalars/halfscalar}
\begin{tikzpicture}
	\begin{pgfonlayer}{nodelayer}
		\node [style=bscalar] (1) at (0, 0) {};
	\end{pgfonlayer}
\end{tikzpicture}
}
\endpgfgraphicnamed.$$
\end{lemma}
\begin{proof} $$%
\beginpgfgraphicnamed{scalars/starroot2prf}
\InputIfFileExists{scalars/starroot2prf.tikz}{}{\input{./figures/scalars/starroot2prf.tikz}}
\endpgfgraphicnamed,$$ where the second equality is obtained using the so-called Hopf law 
\[\qquad\qquad\qquad\qquad\qquad\qquad\qquad%
\beginpgfgraphicnamed{scalars/hopfprf1-bis}
\InputIfFileExists{scalars/hopfprf1-bis.tikz}{}{\input{./figures/scalars/hopfprf1-bis.tikz}}
\endpgfgraphicnamed~=~%
\beginpgfgraphicnamed{scalars/hopfprf6}
\begin{tikzpicture}
	\begin{pgfonlayer}{nodelayer}
		\node [style=gn] (0) at (0, -0.25) {};
		\node [style=none] (1) at (0, -0.75) {};
		\node [style=rn] (2) at (0, 0.25) {};
		\node [style=none] (3) at (0, 0.75) {};
	\end{pgfonlayer}
	\begin{pgfonlayer}{edgelayer}
		\draw (1.center) to (0);
		\draw (2) to (3.center);
	\end{pgfonlayer}
\end{tikzpicture}}
\endpgfgraphicnamed\tag{Hopf}\label{eq:hopf}\] proved for instance in \cite{coecke_interacting_2011}. The third equality is based on the fact that $%
\beginpgfgraphicnamed{scalars/2isroot2square-s}
\InputIfFileExists{scalars/2isroot2square-s.tikz}{}{\input{./figures/scalars/2isroot2square-s.tikz}}
\endpgfgraphicnamed$ which can be proved as follows:
\begin{equation}\qquad\qquad\qquad\qquad\qquad%
\beginpgfgraphicnamed{scalars/2isroot2squarepfnew-s}
\InputIfFileExists{scalars/2isroot2squarepfnew-s.tikz}{}{\input{./figures/scalars/2isroot2squarepfnew-s.tikz}}
\endpgfgraphicnamed\label{eq:double}\end{equation}\end{proof}

By Lemma \ref{starroot2}, it is not necessary to introduce the symbol $%
\beginpgfgraphicnamed{scalars/halfscalar}
\begin{tikzpicture}
	\begin{pgfonlayer}{nodelayer}
		\node [style=bscalar] (1) at (0, 0) {};
	\end{pgfonlayer}
\end{tikzpicture}
}
\endpgfgraphicnamed$, since $%
\beginpgfgraphicnamed{scalars/3linesingle}
\begin{tikzpicture}
	\begin{pgfonlayer}{nodelayer}
		\node [style=rn] (0) at (0, 0.25) {};
		\node [style=gn] (1) at (0, -0.25) {};
	\end{pgfonlayer}
	\begin{pgfonlayer}{edgelayer}
		\draw [bend left=45, looseness=1.00] (0) to (1);
		\draw [bend right=45, looseness=1.00] (0) to (1);
		\draw (0) to (1);
	\end{pgfonlayer}
\end{tikzpicture}}
\endpgfgraphicnamed$ already exists in the \ZX-calculus. However, when removing the $%
\beginpgfgraphicnamed{scalars/halfscalar}
\begin{tikzpicture}
	\begin{pgfonlayer}{nodelayer}
		\node [style=bscalar] (1) at (0, 0) {};
	\end{pgfonlayer}
\end{tikzpicture}
}
\endpgfgraphicnamed$, the one axiom %
\beginpgfgraphicnamed{scalars/star_rule}
\InputIfFileExists{scalars/star_rule.tikz}{}{\input{./figures/scalars/star_rule.tikz}}
\endpgfgraphicnamed using this symbol needs to be treated carefully. Indeed, note that this axiom is necessary for the completeness in the sense that the equation \[\qquad\qquad\qquad\qquad\qquad\qquad\qquad%
\beginpgfgraphicnamed{scalars/dotinverse}
\InputIfFileExists{scalars/dotinverse.tikz}{}{\input{./figures/scalars/dotinverse.tikz}}
\endpgfgraphicnamed\tag{IV'}\label{eq:IV'}\] which is true by combining axiom (SR) and Lemma \ref{starroot2}, cannot be proved without axiom (SR) as there is no other axiom that can be used to transform an empty diagram to a non-empty diagram.

To remove the  $%
\beginpgfgraphicnamed{scalars/halfscalar}
\begin{tikzpicture}
	\begin{pgfonlayer}{nodelayer}
		\node [style=bscalar] (1) at (0, 0) {};
	\end{pgfonlayer}
\end{tikzpicture}
}
\endpgfgraphicnamed$, one could straightforwardly replace all its occurrences -- including in the rewrite rules -- by  $%
\beginpgfgraphicnamed{scalars/3linesingle}
\begin{tikzpicture}
	\begin{pgfonlayer}{nodelayer}
		\node [style=rn] (0) at (0, 0.25) {};
		\node [style=gn] (1) at (0, -0.25) {};
	\end{pgfonlayer}
	\begin{pgfonlayer}{edgelayer}
		\draw [bend left=45, looseness=1.00] (0) to (1);
		\draw [bend right=45, looseness=1.00] (0) to (1);
		\draw (0) to (1);
	\end{pgfonlayer}
\end{tikzpicture}}
\endpgfgraphicnamed  %
\beginpgfgraphicnamed{scalars/3linesingle}
\begin{tikzpicture}
	\begin{pgfonlayer}{nodelayer}
		\node [style=rn] (0) at (0, 0.25) {};
		\node [style=gn] (1) at (0, -0.25) {};
	\end{pgfonlayer}
	\begin{pgfonlayer}{edgelayer}
		\draw [bend left=45, looseness=1.00] (0) to (1);
		\draw [bend right=45, looseness=1.00] (0) to (1);
		\draw (0) to (1);
	\end{pgfonlayer}
\end{tikzpicture}}
\endpgfgraphicnamed$, and consider (\ref{eq:IV'}) as an axiom. We show that the following  simpler axiom can be used instead:

\[\qquad\qquad\qquad\qquad\qquad\qquad\quad\qquad%
\beginpgfgraphicnamed{scalars/inverserule}
\InputIfFileExists{scalars/inverserule.tikz}{}{\input{./figures/scalars/inverserule.tikz}}
\endpgfgraphicnamed\tag{IV}\label{eq:IV}\]

\begin{lemma}\label{rootdot} Given the $%
\beginpgfgraphicnamed{scalars/halfscalar}
\begin{tikzpicture}
	\begin{pgfonlayer}{nodelayer}
		\node [style=bscalar] (1) at (0, 0) {};
	\end{pgfonlayer}
\end{tikzpicture}
}
\endpgfgraphicnamed$-free rules of  Figure \ref{fig:ZX_rules}, (IV) and (IV$'$) are equivalent:
$$%
\beginpgfgraphicnamed{scalars/inverserule}
\InputIfFileExists{scalars/inverserule.tikz}{}{\input{./figures/scalars/inverserule.tikz}}
\endpgfgraphicnamed\quad\Leftrightarrow\quad%
\beginpgfgraphicnamed{scalars/dotinverse}
\InputIfFileExists{scalars/dotinverse.tikz}{}{\input{./figures/scalars/dotinverse.tikz}}
\endpgfgraphicnamed.$$
\end{lemma}
\begin{proof}
$[\Rightarrow]$. Decompose the green dot using \eqref{eq:double}, and then apply the inverse rule \eqref{eq:IV} twice.

$[\Leftarrow]$. We have:
$$%
\beginpgfgraphicnamed{scalars/squaretoinverse}
\InputIfFileExists{scalars/squaretoinverse.tikz}{}{\input{./figures/scalars/squaretoinverse.tikz}}
\endpgfgraphicnamed$$
\end{proof}

\begin{theorem} The $%
\beginpgfgraphicnamed{scalars/halfscalar}
\begin{tikzpicture}
	\begin{pgfonlayer}{nodelayer}
		\node [style=bscalar] (1) at (0, 0) {};
	\end{pgfonlayer}
\end{tikzpicture}
}
\endpgfgraphicnamed$-free rules of  Figure \ref{fig:ZX_rules} together with the inverse rule (IV) are complete for non-zero stabilizer quantum mechanics.
\end{theorem}

\begin{proof}
By completeness of the rules given in Figure  \ref{fig:ZX_rules}, any true equation in stabilizer quantum mechanics can be derived. One can syntactically replace  all occurrences of the $%
\beginpgfgraphicnamed{scalars/halfscalar}
\begin{tikzpicture}
	\begin{pgfonlayer}{nodelayer}
		\node [style=bscalar] (1) at (0, 0) {};
	\end{pgfonlayer}
\end{tikzpicture}
}
\endpgfgraphicnamed$ by $%
\beginpgfgraphicnamed{scalars/3linesingle}
\begin{tikzpicture}
	\begin{pgfonlayer}{nodelayer}
		\node [style=rn] (0) at (0, 0.25) {};
		\node [style=gn] (1) at (0, -0.25) {};
	\end{pgfonlayer}
	\begin{pgfonlayer}{edgelayer}
		\draw [bend left=45, looseness=1.00] (0) to (1);
		\draw [bend right=45, looseness=1.00] (0) to (1);
		\draw (0) to (1);
	\end{pgfonlayer}
\end{tikzpicture}}
\endpgfgraphicnamed  %
\beginpgfgraphicnamed{scalars/3linesingle}
\begin{tikzpicture}
	\begin{pgfonlayer}{nodelayer}
		\node [style=rn] (0) at (0, 0.25) {};
		\node [style=gn] (1) at (0, -0.25) {};
	\end{pgfonlayer}
	\begin{pgfonlayer}{edgelayer}
		\draw [bend left=45, looseness=1.00] (0) to (1);
		\draw [bend right=45, looseness=1.00] (0) to (1);
		\draw (0) to (1);
	\end{pgfonlayer}
\end{tikzpicture}}
\endpgfgraphicnamed$ and get a valid proof using $%
\beginpgfgraphicnamed{scalars/halfscalar}
\begin{tikzpicture}
	\begin{pgfonlayer}{nodelayer}
		\node [style=bscalar] (1) at (0, 0) {};
	\end{pgfonlayer}
\end{tikzpicture}
}
\endpgfgraphicnamed$-free rules and the inverse rule (IV), where each use of the axiom (SR) is replaced by  \eqref{eq:IV'}.
\end{proof}

\subsection{The zero scalar rule is not necessary}

Among all the scalars of the stabilizer \ZX-calculus, those whose interpretation is zero, like \gpi, play a special role. In particular $\gpi$ is an absorbing element: for any diagram $D$, $\gpi \otimes D = \gpi$ should be derivable. In \cite{backens_making_2015}, two axioms are dedicated to zero scalars: the rules (ZO) and (ZS). 
$$%
\beginpgfgraphicnamed{scalars/zo1}
\InputIfFileExists{scalars/zo1.tikz}{}{\input{./figures/scalars/zo1.tikz}}
\endpgfgraphicnamed\quad \text{(ZO)} \qquad \qquad\qquad\qquad     %
\beginpgfgraphicnamed{scalars/zero_scalar}
\begin{tikzpicture}
	\begin{pgfonlayer}{nodelayer}
		\node [style=gn] (0) at (-1, -0) {$~\pi~$};
		\node [style=gn] (1) at (-0.5, -0) {$~\alpha~$};
		\node [style=none] (2) at (0, -0) {$=$};
		\node [style=gn] (3) at (0.5, -0) {$~\pi~$};
	\end{pgfonlayer}
\end{tikzpicture}}
\endpgfgraphicnamed\quad \text{(ZS)}$$
Intuitively, one can derive  $\gpi \otimes D = \gpi$ for an arbitrary $D$ as follows: first the (ZO) rule is used to `cut' all wires of $D$, then the remaining pieces are absorbed using the (ZS) rule. 

We show that the (ZS) rule $%
\beginpgfgraphicnamed{scalars/zero_scalar}
\begin{tikzpicture}
	\begin{pgfonlayer}{nodelayer}
		\node [style=gn] (0) at (-1, -0) {$~\pi~$};
		\node [style=gn] (1) at (-0.5, -0) {$~\alpha~$};
		\node [style=none] (2) at (0, -0) {$=$};
		\node [style=gn] (3) at (0.5, -0) {$~\pi~$};
	\end{pgfonlayer}
\end{tikzpicture}}
\endpgfgraphicnamed$ is not necessary for completeness and can be derived from the other rules of the language.
To  this end, we first show that even without the presence of the absorbing element $\gpi$, some  angles can be removed:

\begin{lemma}\label{alphadelete} For any $\alpha$, $~%
\beginpgfgraphicnamed{scalars/alphadelete}
\begin{tikzpicture}
	\begin{pgfonlayer}{nodelayer}
		\node [style=rn] (0) at (0.2, 0.25) {};
		\node [style=gn] (1) at (0.2, -0.25) {$\alpha$};
		\node [style=none] (2) at (0.5, -0) {$=$};
		\node [style=rn] (3) at (0.8, 0.25) {};
		\node [style=gn] (4) at (0.8, -0.25) {};
	\end{pgfonlayer}
	\begin{pgfonlayer}{edgelayer}
		\draw (0) to (1);
		\draw (3) to (4);
	\end{pgfonlayer}
\end{tikzpicture}}
\endpgfgraphicnamed~$ is derivable without using the (ZS) rule.
\end{lemma}

\begin{proof} First, note that the equation for $\alpha=\pi$ can be derived directly from the (K1) rule instantiated with no output:
$$%
\beginpgfgraphicnamed{scalars/k1-0-proof}
\InputIfFileExists{scalars/k1-0-proof.tikz}{}{\input{./figures/scalars/k1-0-proof.tikz}}
\endpgfgraphicnamed$$
Then for arbitrary $\alpha$, we have:
$$%
\beginpgfgraphicnamed{scalars/alphadeleteproof}
\InputIfFileExists{scalars/alphadeleteproof.tikz}{}{\input{./figures/scalars/alphadeleteproof.tikz}}
\endpgfgraphicnamed$$
\end{proof}

\begin{theorem}\label{zerotosaclarzerolm} For any $\alpha$, $%
\beginpgfgraphicnamed{scalars/zero_scalar}
\begin{tikzpicture}
	\begin{pgfonlayer}{nodelayer}
		\node [style=gn] (0) at (-1, -0) {$~\pi~$};
		\node [style=gn] (1) at (-0.5, -0) {$~\alpha~$};
		\node [style=none] (2) at (0, -0) {$=$};
		\node [style=gn] (3) at (0.5, -0) {$~\pi~$};
	\end{pgfonlayer}
\end{tikzpicture}}
\endpgfgraphicnamed$ is derivable without using the (ZS) rule.
\end{theorem}

\begin{proof}
First, we show that~ $%
\beginpgfgraphicnamed{scalars/piroot2}
\begin{tikzpicture}
	\begin{pgfonlayer}{nodelayer}
		\node [style=gn] (0) at (0, -0) {$\pi$};
		\node [style=rn] (1) at (0.5, 0.25) {};
		\node [style=gn] (2) at (0.5, -0.25) {};
		\node [style=none] (3) at (1, -0) {$=$};
		\node [style=gn] (4) at (1.5, -0) {$\pi$};
	\end{pgfonlayer}
	\begin{pgfonlayer}{edgelayer}
		\draw (1) to (2);
	\end{pgfonlayer}
\end{tikzpicture}}
\endpgfgraphicnamed$~: 
$$
\beginpgfgraphicnamed{scalars/piabsorbroot2pf}
\InputIfFileExists{scalars/piabsorbroot2pf.tikz}{}{\input{./figures/scalars/piabsorbroot2pf.tikz}}
\endpgfgraphicnamed.
$$ Now, for any $\alpha$, 
\begin{align*}
\beginpgfgraphicnamed{scalars/piabsorbalphaproofzerorule}
\InputIfFileExists{scalars/piabsorbalphaproofzerorule.tikz}{}{\input{./figures/scalars/piabsorbalphaproofzerorule.tikz}}
\endpgfgraphicnamed.
\end{align*}
\end{proof}

\begin{remark}The proof that $%
\beginpgfgraphicnamed{scalars/zero_scalar}
\begin{tikzpicture}
	\begin{pgfonlayer}{nodelayer}
		\node [style=gn] (0) at (-1, -0) {$~\pi~$};
		\node [style=gn] (1) at (-0.5, -0) {$~\alpha~$};
		\node [style=none] (2) at (0, -0) {$=$};
		\node [style=gn] (3) at (0.5, -0) {$~\pi~$};
	\end{pgfonlayer}
\end{tikzpicture}}
\endpgfgraphicnamed$ can be derived from the other rules of the language is not specific to the stabilizer case. In the full \ZX-calculus, i.e.\ for arbitrary angles $\alpha \in [0, 2\pi)$, $%
\beginpgfgraphicnamed{scalars/zero_scalar}
\begin{tikzpicture}
	\begin{pgfonlayer}{nodelayer}
		\node [style=gn] (0) at (-1, -0) {$~\pi~$};
		\node [style=gn] (1) at (-0.5, -0) {$~\alpha~$};
		\node [style=none] (2) at (0, -0) {$=$};
		\node [style=gn] (3) at (0.5, -0) {$~\pi~$};
	\end{pgfonlayer}
\end{tikzpicture}}
\endpgfgraphicnamed$ can be derived from the other rules of the language, too. 
\end{remark}

\begin{corollary}
The rules of Figure \ref{fig:ZX_rules} without (ZS) and with (SR) replaced by (IV) are complete for stabilizer quantum mechanics.
\end{corollary}

\subsection{Minimality of the scalar axioms}
\label{s:scalar_minimality}

From the  rules and symbol dedicated to scalars in Figure \ref{fig:ZX_rules}, we have eliminated one symbol and one rule, with two scalar rules remaining. We now show that this set of rules is optimal for scalars in the sense that both of those axioms are necessary. Indeed, the inverse rule (IV) cannot be derived using the other rules since it is the only rule which equates an empty diagram and a non empty diagram. To prove that the zero rule (ZO) is also necessary, we introduce an alternative interpretation of the diagrams, which is sound for all the rules of the language except for the zero rule.

\begin{theorem} \label{zononderive}
The (ZO) rule  $$%
\beginpgfgraphicnamed{scalars/zo1}
\InputIfFileExists{scalars/zo1.tikz}{}{\input{./figures/scalars/zo1.tikz}}
\endpgfgraphicnamed$$ cannot be derived from the other rules of Figure \ref{fig:ZX_rules} without (ZS) and with (SR) replaced by (IV).
\end{theorem}

\begin{proof}
Consider an alternative interpretation $\llbracket \cdot \rrbracket^{\natural}$ of stabilizer \ZX-calculus diagrams: for any diagram $D: n\to m$, let $\den{D}: 2n\to 2m$ be a diagram defined as follows:

\centerline{$
\def\arraystretch{0.5}
\denote{D_1\otimes D_2}^\natural := \denote{D_1}^\natural{\otimes} \denote{D_2}^\natural\quad  ~\denote{D_2\circ D_1}^\natural := \denote{D_2}^\natural{\circ}\denote{D_1}^\natural\quad~\den{~\dempty~} := \dempty \quad~ \den{~\didf~}:=~\did~~~\did$}

\centerline{$
\def\arraystretch{0.5}
\den{~\dHf~} := ~\dsigmaf~\qquad\den{~\dsigmaf~}:= ~\dsigmaf\!\!\!\!\!\!\!\dsigmaf~\qquad \den{\drcupf}:=  ~\drcup\!\!\!\!\!\!\!\!\!\!\!\!\!\!\!\!\!\drcup~ \qquad \den{\drcapf}:= ~\drcap\!\!\!\!\!\!\!\!\!\!\!\!\!\!\!\!\!\drcap~ $}

~\\
\noindent \begin{tabular}{lccc}
 When $\alpha=0\bmod \pi$,& $\den{~%
\beginpgfgraphicnamed{scalars-s/spiderg-a}
\InputIfFileExists{scalars-s/spiderg-a.tikz}{}{\input{./figures/scalars-s/spiderg-a.tikz}}
\endpgfgraphicnamed~}  ~:=~ %
\beginpgfgraphicnamed{scalars-s/spiderinterpretationsc-2}
\InputIfFileExists{scalars-s/spiderinterpretationsc-2.tikz}{}{\input{./figures/scalars-s/spiderinterpretationsc-2.tikz}}
\endpgfgraphicnamed$&\quad&$\den{~%
\beginpgfgraphicnamed{scalars-s/spiderr-a}
\InputIfFileExists{scalars-s/spiderr-a.tikz}{}{\input{./figures/scalars-s/spiderr-a.tikz}}
\endpgfgraphicnamed~}  ~:=~ %
\beginpgfgraphicnamed{scalars-s/rspiderinterpretationsc-2}
\InputIfFileExists{scalars-s/rspiderinterpretationsc-2.tikz}{}{\input{./figures/scalars-s/rspiderinterpretationsc-2.tikz}}
\endpgfgraphicnamed$\\&&&\\
  When $\alpha=\frac \pi 2\bmod \pi$,& $\den{~%
\beginpgfgraphicnamed{scalars-s/spiderg-a}
\InputIfFileExists{scalars-s/spiderg-a.tikz}{}{\input{./figures/scalars-s/spiderg-a.tikz}}
\endpgfgraphicnamed~}  ~:=~ %
\beginpgfgraphicnamed{scalars-s/spiderinterpretationsc-2-cc}
\InputIfFileExists{scalars-s/spiderinterpretationsc-2-cc.tikz}{}{\input{./figures/scalars-s/spiderinterpretationsc-2-cc.tikz}}
\endpgfgraphicnamed$&&$\den{~%
\beginpgfgraphicnamed{scalars-s/spiderr-a}
\InputIfFileExists{scalars-s/spiderr-a.tikz}{}{\input{./figures/scalars-s/spiderr-a.tikz}}
\endpgfgraphicnamed~}  ~:=~ %
\beginpgfgraphicnamed{scalars-s/rspiderinterpretationsc-2-cc}
\InputIfFileExists{scalars-s/rspiderinterpretationsc-2-cc.tikz}{}{\input{./figures/scalars-s/rspiderinterpretationsc-2-cc.tikz}}
\endpgfgraphicnamed$\\
 \end{tabular}

~\\
Basically, $\den D$ is an angle-free encoding of $D$, where each dot is encoded by two dots, one of each colour. When the angle is $\pm\frac \pi 2$ the two copies are connected. Note that the encodings of $\alpha$ and $\alpha +\pi$ are the same.  The interpretation of the Hadamard node, $\den \dHf$, swaps the two sides of the encoding.  It is a routine check that $\llbracket \cdot \rrbracket^{\natural}$ is a sound interpretation for all rules in Figure \ref{fig:ZX_rules} with (SR) replaced by (IV), except for (ZO). Here, we just give the most interesting case, that of the Euler decomposition rule (EU):
 $$ \begin{array}{lll}
 \left\llbracket ~%
\beginpgfgraphicnamed{scalars/righthaddecom}
\InputIfFileExists{scalars/righthaddecom.tikz}{}{\input{./figures/scalars/righthaddecom.tikz}}
\endpgfgraphicnamed~\right\rrbracket^{\natural}&~=~&%
\beginpgfgraphicnamed{scalars/doubleinthadprf1}
\InputIfFileExists{scalars/doubleinthadprf1.tikz}{}{\input{./figures/scalars/doubleinthadprf1.tikz}}
\endpgfgraphicnamed~~=~~%
\beginpgfgraphicnamed{scalars/doubleinthadprf2}
\InputIfFileExists{scalars/doubleinthadprf2.tikz}{}{\input{./figures/scalars/doubleinthadprf2.tikz}}
\endpgfgraphicnamed\vspace{0.5cm}\\
 &=&%
\beginpgfgraphicnamed{scalars/doubleinthadprf3}
\InputIfFileExists{scalars/doubleinthadprf3.tikz}{}{\input{./figures/scalars/doubleinthadprf3.tikz}}
\endpgfgraphicnamed~~=~~%
\beginpgfgraphicnamed{scalars/doubleinthadprf4}
\InputIfFileExists{scalars/doubleinthadprf4.tikz}{}{\input{./figures/scalars/doubleinthadprf4.tikz}}
\endpgfgraphicnamed~~=~~\left\llbracket ~ %
\beginpgfgraphicnamed{scalars/lefthaddecom}
\InputIfFileExists{scalars/lefthaddecom.tikz}{}{\input{./figures/scalars/lefthaddecom.tikz}}
\endpgfgraphicnamed~\right\rrbracket^{\natural}
 \end{array}$$
The first equation is simply the definition of $\den{\cdot}$. The second step consists of applying the (B2) rule on the top right of the diagram. The third step involves ($i$) transforming the  `square' into a single green dot, and ($ii$)  using the (S1) rule to merge adjacent same-colour dots. Finally, we apply the Hopf law, leading to the encoding of the RHS of the (EU) rule. 
 
 If (ZO) could be derived from the other rules of the language, then soundness of $\den{\cdot}$ would imply that:
 \begin{equation}\label{eq:zerorulenatural}
  \left\llbracket~ %
\beginpgfgraphicnamed{scalars/zeroruledisconect}
\begin{tikzpicture}
	\begin{pgfonlayer}{nodelayer}
		\node [style=gn] (0) at (0.5, 0) {$\pi$};
		\node [style=none] (1) at (1, -0.75) {};
		\node [style=rn] (2) at (1, -0.25) {};
		\node [style=none] (3) at (1, 0.75) {};
		\node [style=gn] (4) at (1, 0.25) {};
	\end{pgfonlayer}
	\begin{pgfonlayer}{edgelayer}
		\draw (3.center) to (4);
		\draw (2) to (1.center);
	\end{pgfonlayer}
\end{tikzpicture}}
\endpgfgraphicnamed~\right\rrbracket^{\natural} =  \left\llbracket~ %
\beginpgfgraphicnamed{scalars/zeroruleconnect}
\begin{tikzpicture}
	\begin{pgfonlayer}{nodelayer}
		\node [style=none] (0) at (0.75, 0.5) {};
		\node [style=gn] (1) at (0.5, 0) {$\pi$};
		\node [style=none] (2) at (0.75, -0.5) {};
	\end{pgfonlayer}
	\begin{pgfonlayer}{edgelayer}
		\draw (0.center) to (2.center);
	\end{pgfonlayer}
\end{tikzpicture}}
\endpgfgraphicnamed~\right\rrbracket^{\natural}.
 \end{equation}
 On the other hand, by applying $\den{\cdot}$ to the two sides of (ZO), we find that:
 \[
  \left\llbracket~ %
\beginpgfgraphicnamed{scalars/zeroruledisconect}
\begin{tikzpicture}
	\begin{pgfonlayer}{nodelayer}
		\node [style=gn] (0) at (0.5, 0) {$\pi$};
		\node [style=none] (1) at (1, -0.75) {};
		\node [style=rn] (2) at (1, -0.25) {};
		\node [style=none] (3) at (1, 0.75) {};
		\node [style=gn] (4) at (1, 0.25) {};
	\end{pgfonlayer}
	\begin{pgfonlayer}{edgelayer}
		\draw (3.center) to (4);
		\draw (2) to (1.center);
	\end{pgfonlayer}
\end{tikzpicture}}
\endpgfgraphicnamed~\right\rrbracket^{\natural}=~~%
\beginpgfgraphicnamed{scalars/doubleintzeroruledisconect}
\InputIfFileExists{scalars/doubleintzeroruledisconect.tikz}{}{\input{./figures/scalars/doubleintzeroruledisconect.tikz}}
\endpgfgraphicnamed \qquad \qquad  \text {and} \qquad \qquad 
 \left\llbracket~ %
\beginpgfgraphicnamed{scalars/zeroruleconnect}
\begin{tikzpicture}
	\begin{pgfonlayer}{nodelayer}
		\node [style=none] (0) at (0.75, 0.5) {};
		\node [style=gn] (1) at (0.5, 0) {$\pi$};
		\node [style=none] (2) at (0.75, -0.5) {};
	\end{pgfonlayer}
	\begin{pgfonlayer}{edgelayer}
		\draw (0.center) to (2.center);
	\end{pgfonlayer}
\end{tikzpicture}}
\endpgfgraphicnamed~\right\rrbracket^{\natural}=~~%
\beginpgfgraphicnamed{scalars/doubleintzeroruleconnect}
\InputIfFileExists{scalars/doubleintzeroruleconnect.tikz}{}{\input{./figures/scalars/doubleintzeroruleconnect.tikz}}
\endpgfgraphicnamed
 \]
 i.e.\ the left-hand side of \eqref{eq:zerorulenatural} and the right-hand side of \eqref{eq:zerorulenatural} do not have the same standard interpretation. Hence, by soundness of $\denote{\cdot}$ and $\den{\cdot}$, \eqref{eq:zerorulenatural} cannot be derived from the other rules. 
 
 Therefore, the (ZO) rule is not derivable in the stabilizer \ZX-calculus.
 \end{proof}

To sum up: we have removed the rules (SR) and (ZS), and replaced them with the inverse rule (IV). The remaining rules about scalars -- (IV) and (ZO) -- have been shown not to be derivable from the other rules of the \ZX-calculus. They therefore form a set of minimal axioms for scalars in the stabilizer \ZX-calculus.

\section{Simplified Rules}
\label{s:simplified}

In this section, we give a new system of rules for the stabilizer \ZX-calculus, shown in Figure \ref{figure3}. The new rules are much simpler than those given in Figure \ref{fig:ZX_rules}, while being just as powerful: the new set of rules can be proved to be equivalent to the old ones without any reference to the convention that any rule also holds with the colours red and green swapped or with diagrams flipped upside-down, yet has the advantage of possessing fewer equalities.

\begin{figure}[h]
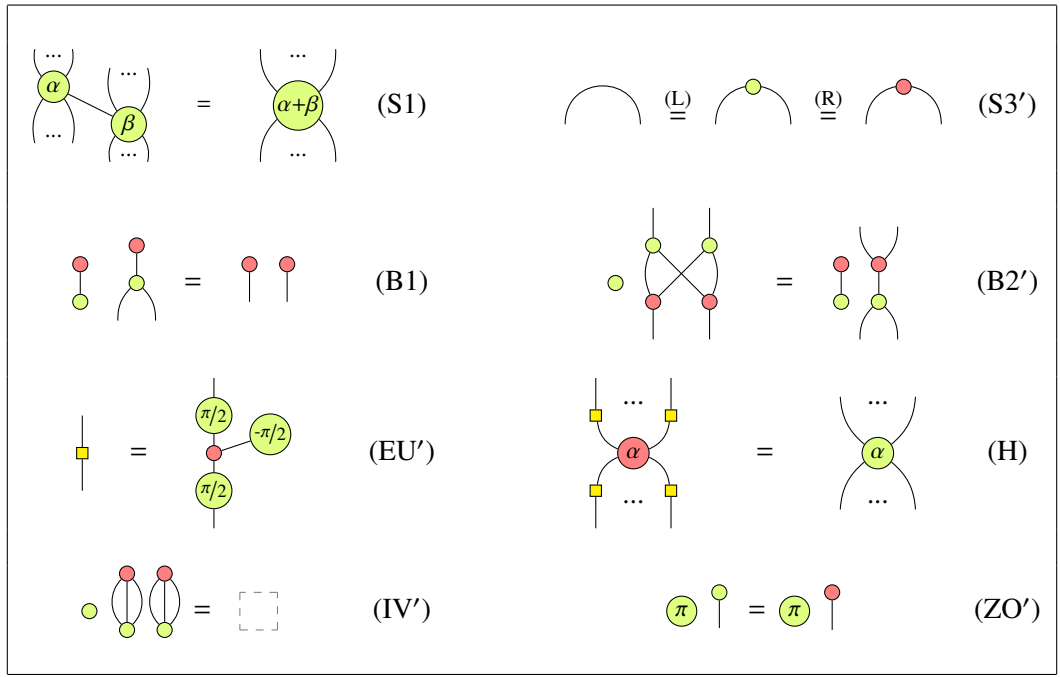

 \centering
 \begin{tabular}{|ccccc|}
  \hline
  &&&& \\
\beginpgfgraphicnamed{scalars/spider-bis}
\InputIfFileExists{scalars/spider-bis.tikz}{}{\input{./figures/scalars/spider-bis.tikz}}
\endpgfgraphicnamed&(S1)&$\qquad$&%
\beginpgfgraphicnamed{scalars/induced_compact_structure}
\InputIfFileExists{scalars/induced_compact_structure.tikz}{}{\input{./figures/scalars/induced_compact_structure.tikz}}
\endpgfgraphicnamed&(S3$'$)\\
  &&&& \\
\beginpgfgraphicnamed{scalars/b1s}
\InputIfFileExists{scalars/b1s.tikz}{}{\input{./figures/scalars/b1s.tikz}}
\endpgfgraphicnamed&(B1)&&%
\beginpgfgraphicnamed{scalars/b2snew}
\InputIfFileExists{scalars/b2snew.tikz}{}{\input{./figures/scalars/b2snew.tikz}}
\endpgfgraphicnamed&(B2$'$)\\
  &&&& \\
\beginpgfgraphicnamed{scalars/HadaDecomSingles-prime}
\InputIfFileExists{scalars/HadaDecomSingles-prime.tikz}{}{\input{./figures/scalars/HadaDecomSingles-prime.tikz}}
\endpgfgraphicnamed&(EU$'$)&&%
\beginpgfgraphicnamed{scalars/h2}
\InputIfFileExists{scalars/h2.tikz}{}{\input{./figures/scalars/h2.tikz}}
\endpgfgraphicnamed&(H)\\
  &&&& \\
\beginpgfgraphicnamed{scalars/dotinverse}
\InputIfFileExists{scalars/dotinverse.tikz}{}{\input{./figures/scalars/dotinverse.tikz}}
\endpgfgraphicnamed &(IV$'$)&&%
\beginpgfgraphicnamed{scalars/zo1-prime}
\InputIfFileExists{scalars/zo1-prime.tikz}{}{\input{./figures/scalars/zo1-prime.tikz}}
\endpgfgraphicnamed&(ZO$'$)\\
  &&&& \\
  \hline
 \end{tabular}
 \caption{Simplified rules for the stabilizer \ZX-calculus, still using the conventions that the right-hand side of (IV$'$) is an empty diagram and that ellipses denote zero or more wires.}\label{figure3}
\end{figure}

\subsection{Soundness of the simplified rule set}
\label{s:soundness_simplified}

To show that the rewrite rules given in Figure \ref{figure3} are a valid rule system for the \ZX-calculus, we first need to show that they are sound. The rewrite rules in Figure \ref{fig:ZX_rules} are known to be sound \cite{coecke_interacting_2011,backens_making_2015}, thus to prove soundness of the new system it is sufficient to show that all new rules can be derived from the old ones.

\begin{theorem}
 The rewrite rules given in Figure \ref{figure3} are sound.
\end{theorem}
\begin{proof}
 The rules (S1), (S3$'$), (B1) and (H) of Figure \ref{figure3} already exist in Figure \ref{fig:ZX_rules}, with (S3$'$) containing both the colour-swapped and the original version of (S3). Thus, to prove that the new system is sound, it suffices to show that the remaining rules -- (B2$'$), (EU$'$), (IV$'$) and (ZO$'$) -- can be derived from the rules in Figure \ref{fig:ZX_rules}.
 \begin{itemize}
  \item Soundness of rule (B2$'$) follows immediately from (B2) via the fact that $%
\beginpgfgraphicnamed{scalars/2isroot2square-s}
\InputIfFileExists{scalars/2isroot2square-s.tikz}{}{\input{./figures/scalars/2isroot2square-s.tikz}}
\endpgfgraphicnamed$, as shown in the proof of Lemma \ref{starroot2}.
  \item For soundness of rule (EU$'$), note that:
   \begin{equation}\label{eq:Y-state}
\beginpgfgraphicnamed{scalars/hadsimplifyprfirst}
\InputIfFileExists{scalars/hadsimplifyprfirst.tikz}{}{\input{./figures/scalars/hadsimplifyprfirst.tikz}}
\endpgfgraphicnamed,
   \end{equation}
   where we have used the colour-swapped version of (EU), (IV), (B1) and Lemma \ref{alphadelete}.

   Then:
   \begin{align*}
\beginpgfgraphicnamed{scalars/hadsimplifyprf1}
\InputIfFileExists{scalars/hadsimplifyprf1.tikz}{}{\input{./figures/scalars/hadsimplifyprf1.tikz}}
\endpgfgraphicnamed= %
\beginpgfgraphicnamed{scalars/hadsimplifyprf2}
\InputIfFileExists{scalars/hadsimplifyprf2.tikz}{}{\input{./figures/scalars/hadsimplifyprf2.tikz}}
\endpgfgraphicnamed= %
\beginpgfgraphicnamed{scalars/hadsimplifyprf3}
\InputIfFileExists{scalars/hadsimplifyprf3.tikz}{}{\input{./figures/scalars/hadsimplifyprf3.tikz}}
\endpgfgraphicnamed=%
\beginpgfgraphicnamed{scalars/Had2}
\begin{tikzpicture}
	\begin{pgfonlayer}{nodelayer}
		\node [style={H box}] (0) at (0.5, 0) {};
		\node [style=none] (1) at (0.5, 0.5) {};
		\node [style=none] (2) at (0.5, -0.5) {};
	\end{pgfonlayer}
	\begin{pgfonlayer}{edgelayer}
		\draw (1.center) to (0);
		\draw (2.center) to (0);
	\end{pgfonlayer}
\end{tikzpicture}}
\endpgfgraphicnamed
   \end{align*}
   using \eqref{eq:Y-state}, the colour-swapped version of (S1), (EU), and (IV).
  \item Soundness of rule (IV$'$) follows immediately from Lemma \ref{rootdot}.
  \item For soundness of rule (ZO$'$), start from the left-hand side and apply (ZO) followed by Theorem \ref{zerotosaclarzerolm} to find:
   \begin{align*}
\beginpgfgraphicnamed{scalars/zosimproofnew}
\InputIfFileExists{scalars/zosimproofnew.tikz}{}{\input{./figures/scalars/zosimproofnew.tikz}}
\endpgfgraphicnamed
   \end{align*}
 \end{itemize}
 This completes the proof.
\end{proof}

\subsection{Completeness of the simplified rule set}

If we want to replace the rules in Figure \ref{fig:ZX_rules} with those in Figure \ref{figure3}, we also need to be able to derive any rule in the former -- including their colour-swapped and upside-down variants -- from the latter. As the system in Figure \ref{fig:ZX_rules} is known to be complete \cite{backens_zx-calculus_2013,backens_making_2015}, this is equivalent to proving completeness for the new rule system.

\begin{theorem}\label{thm:completeness_with_topology}
 The rewrite rules given in Figure \ref{figure3}, together with the topology meta rule, are complete for the stabilizer \ZX-calculus.
\end{theorem}
\begin{proof}
 We derive the rules in Figure \ref{fig:ZX_rules} one-by-one, each together with its colour-swapped and upside-down versions, in the following order: (H), (S1) \& (S3), (B1), (IV), (B2), (S1$'$), (K1), (EU), (K2), (ZO). Later proofs rely on earlier ones. These derivations are given as separate lemmas in Appendix \ref{a:completeness_proof}.
 
 As an example, we provide a sketch derivation of \eqref{eq:k2sketch}, which is the rule (K2) with $\alpha=\frac{\pi}{2}$:
 \begin{equation}\label{eq:k2sketch}
\beginpgfgraphicnamed{scalars/k2sketch}
\InputIfFileExists{scalars/k2sketch.tikz}{}{\input{./figures/scalars/k2sketch.tikz}}
\endpgfgraphicnamed
 \end{equation}
 To prove this, we first need to derive some relationships between the states with phases $\frac{\pi}{2}$ and $-\frac{\pi}{2}$ in the two colours. These equations, given below, essentially follow from (EU) (see Lemma \ref{piby2transformation}):
 \begin{equation*}\label{eq:pi2}
\beginpgfgraphicnamed{scalars/piby2trans}
\InputIfFileExists{scalars/piby2trans.tikz}{}{\input{./figures/scalars/piby2trans.tikz}}
\endpgfgraphicnamed,\qquad\quad\quad %
\beginpgfgraphicnamed{scalars/piby2transcolorch}
\InputIfFileExists{scalars/piby2transcolorch.tikz}{}{\input{./figures/scalars/piby2transcolorch.tikz}}
\endpgfgraphicnamed,\qquad\quad\quad %
\beginpgfgraphicnamed{scalars/piby2redtowhite}
\InputIfFileExists{scalars/piby2redtowhite.tikz}{}{\input{./figures/scalars/piby2redtowhite.tikz}}
\endpgfgraphicnamed.
\end{equation*}
 Next, we show that %
\beginpgfgraphicnamed{scalars/piby2saclarmultiplys}
\InputIfFileExists{scalars/piby2saclarmultiplys.tikz}{}{\input{./figures/scalars/piby2saclarmultiplys.tikz}}
\endpgfgraphicnamed (cf.\ Corollary \ref{cor:piby2scalars}): 
\begin{align*}
\beginpgfgraphicnamed{scalars/piby2saclarmultiplyprf2}
\InputIfFileExists{scalars/piby2saclarmultiplyprf2.tikz}{}{\input{./figures/scalars/piby2saclarmultiplyprf2.tikz}}
\endpgfgraphicnamed
\end{align*}
Finally, the $\frac{\pi}{2}$ case of (K2) can be derived as follows:
\begin{align*}
\beginpgfgraphicnamed{scalars/k2ruleprf}
\InputIfFileExists{scalars/k2ruleprf.tikz}{}{\input{./figures/scalars/k2ruleprf.tikz}}
\endpgfgraphicnamed
 \end{align*}
 Details of this derivation can be found in Lemma \ref{k2rule}.
\end{proof}

Note that the derivation of (K2) only works within the stabilizer fragment of the \ZX-calculus, i.e.\ for phases that are integer multiples of $\frac{\pi}{2}$; if more general phase angles are allowed, (K2) may still be necessary.

\section{Conclusion and perspectives}

The stabilizer \ZX-calculus has a complete set of rewrite rules, which allow any equality that can be derived using matrices to also be derived graphically.
We introduce a simplified but still complete version of the stabilizer \ZX-calculus with significantly fewer rewrite rules. In particular, many rules obtained from others by swapping colours and/or flipping diagrams upside-down are no longer assumed. Our aim is to minimise the axioms of the language in order to pinpoint the fundamental structures of quantum mechanics, and also simplify the development and the efficiency of automated  tools for quantum reasoning, like Quantomatic \cite{quanto}.

The simplified stabilizer \ZX-calculus can also serve as a backbone for further developments, in particular concerning other fragments, like the real stabilizer \ZX-calculus \cite{duncan_pivoting_2014} or the full calculus (allowing arbitrary angles). Several rules we showed to be derivable in the stabilizer \ZX-calculus are also derivable in the full \ZX-calculus: e.g.\ (ZS), which is valid for arbitrary angles, and (K1). The derivation of (K2) on the other hand is valid for the stabilizer fragment only.  Recently, a new rule called supplementarity was proved to be necessary for the (full) \ZX-calculus \cite{PW15}, and in particular  for the $\pi/4$-fragment of the \ZX-calculus, which corresponds to the so called Clifford+T quantum mechanics. Even if supplementarity and (K2) rules can be derived in the stabilizer \ZX-calculus, a future project is to establish a simple, possibly minimal, set of axioms for the stabilizer \ZX-calculus which contains the rules known to be necessary for arbitrary angles (like supplementarity or (K2)), while avoiding rules which are in some sense specific to the $\pi/2$ fragment, e.g.\ (EU).

\subparagraph*{Acknowledgements}

The authors would like to thank Bob Coecke, Ross Duncan, Emmanuel Jeandel, and Aleks Kissinger for valuable discussions. MB acknowledges funding from EPSRC,  QW acknowledges funding from R\'egion Lorraine.

%%
%% Bibliography
%%

%% Either use bibtex (recommended),

\bibliographystyle{eptcs}
\bibliography{refs}

%% .. or use the thebibliography environment explicitely

\appendix
\section{Completeness proof for the simplified rule set}
\label{a:completeness_proof}

First note that, because of the principle that `only the topology matters', the upside-down versions of any rewrite rules can be derived by simply applying cups or caps to all outputs or inputs.

Now, to derive the colour-swapped version of (H), it is useful to first show that the Hadamard node is self-inverse.

\begin{lemma}\label{rh2}
The Hadamard node is self-inverse:
\begin{equation}\label{eq:Hadamard_self-inv}
\beginpgfgraphicnamed{scalars/hadsquare}
\InputIfFileExists{scalars/hadsquare.tikz}{}{\input{./figures/scalars/hadsquare.tikz}}
\endpgfgraphicnamed
 \end{equation}
\end{lemma}
\begin{proof}
As `only the topology matters', we can use (S3$'$) directly, together with the colour change rule, (H):
\begin{align*}
\beginpgfgraphicnamed{scalars/hadsquareprf0}
\begin{tikzpicture}
	\begin{pgfonlayer}{nodelayer}
		\node [style=none] (0) at (0, 0.75) {};
		\node [style={H box}] (1) at (0, 0.25) {};
		\node [style={H box}] (2) at (0, -0.25) {};
		\node [style=none] (3) at (0, -0.75) {};
	\end{pgfonlayer}
	\begin{pgfonlayer}{edgelayer}
		\draw (3.center) to (2);
		\draw (0.center) to (2);
	\end{pgfonlayer}
\end{tikzpicture}}
\endpgfgraphicnamed=%
\beginpgfgraphicnamed{scalars/hadsquareprf1}
\InputIfFileExists{scalars/hadsquareprf1.tikz}{}{\input{./figures/scalars/hadsquareprf1.tikz}}
\endpgfgraphicnamed=%
\beginpgfgraphicnamed{scalars/hadsquareprf2}
\InputIfFileExists{scalars/hadsquareprf2.tikz}{}{\input{./figures/scalars/hadsquareprf2.tikz}}
\endpgfgraphicnamed
 \end{align*}
This yields the desired result.
\end{proof}

\begin{lemma}\label{hswap}
  The colour-swapped version of (H) can be derived:
\begin{equation}\label{eq:H_swap}
\beginpgfgraphicnamed{scalars/hswap}
\InputIfFileExists{scalars/hswap.tikz}{}{\input{./figures/scalars/hswap.tikz}}
\endpgfgraphicnamed
 \end{equation}
\end{lemma}
\begin{proof}
This follows immediately by applying Hadamard nodes to all inputs and outputs of (H) and using \eqref{eq:Hadamard_self-inv}.
\end{proof}

Given the colour-swapped version of (H) and the self-inverse property of the Hadamard node, colour-swapped versions of any rewrite rules can be derived by applying Hadamards to all inputs and outputs, using (H) and \eqref{eq:H_swap}, and \eqref{eq:Hadamard_self-inv}. We will not state this explicitly for each separate rule.

Note that being able to derive colour-swapped and upside-down versions of any rewrite rule means we do not need to do any more work to derive (S1), (S3), (B1) and (H) in the form in which they appear in Figure \ref{fig:ZX_rules}.

\begin{lemma}[Hopf law]\label{hopflaw}
 The Hopf law holds, i.e.:
 \begin{equation}\label{eq:Hopf}
\beginpgfgraphicnamed{scalars/hopf}
\InputIfFileExists{scalars/hopf.tikz}{}{\input{./figures/scalars/hopf.tikz}}
\endpgfgraphicnamed
 \end{equation}
\end{lemma}
\begin{proof}
 By the topology meta rule, we can introduce a twist in one of the wires connecting the red and green dots on the left-hand side of the above. Then:
 \begin{align*}
\beginpgfgraphicnamed{scalars/hopfproof}
\InputIfFileExists{scalars/hopfproof.tikz}{}{\input{./figures/scalars/hopfproof.tikz}}
\endpgfgraphicnamed
 \end{align*}
 where we used (S3$'$), (S1), the red spider rule and (B2$'$).
\end{proof}

\begin{lemma}\label{dotdecom}
A  dot can be decomposed:
\begin{equation}\label{eq:dotdecom}
\beginpgfgraphicnamed{scalars/dotdecomp}
\InputIfFileExists{scalars/dotdecomp.tikz}{}{\input{./figures/scalars/dotdecomp.tikz}}
\endpgfgraphicnamed
 \end{equation}
 \end{lemma}
 \begin{proof}
 Starting from the right-hand side, use (IV$'$), (S1), the Hopf law, the colour-swapped version of (B1), (S3$'$), \eqref{eq:Hadamard_self-inv}, \eqref{eq:H_swap} and the red spider rule. Then:
 \begin{align*}
\beginpgfgraphicnamed{scalars/dotdecompoof}
\InputIfFileExists{scalars/dotdecompoof.tikz}{}{\input{./figures/scalars/dotdecompoof.tikz}}
\endpgfgraphicnamed
  \end{align*}
 This completes the proof.
 \end{proof}

\begin{lemma}\label{inverse}
 The inverse rule (IV) can be derived:
 \begin{equation}\label{eq:iv}
\beginpgfgraphicnamed{scalars/inverserule}
\InputIfFileExists{scalars/inverserule.tikz}{}{\input{./figures/scalars/inverserule.tikz}}
\endpgfgraphicnamed
 \end{equation}
\end{lemma}
\begin{proof}
The proof is the same as the $[\Leftarrow]$ part proof of 
Lemma \ref{rootdot}, using (IV$'$), the Hopf law and \eqref{eq:dotdecom}.
\end{proof}

\begin{lemma}\label{b2}
(B2) can be derived:
 \begin{equation}\label{eq:B2}
\beginpgfgraphicnamed{scalars/b2s}
\InputIfFileExists{scalars/b2s.tikz}{}{\input{./figures/scalars/b2s.tikz}}
\endpgfgraphicnamed
 \end{equation}
\end{lemma}
\begin{proof}
This follows immediately from (B2$'$) via  \eqref{eq:dotdecom} and \eqref{eq:iv}.
\end{proof}

Since (IV) and (B2) have been derived from the rules listed in Figure \ref{figure3}, we will use them in the future for the derivation of other rules in Figure \ref{fig:ZX_rules}.

Swapping the colours in (B2) has the same effect as flipping it upside-down, so this variant follows immediately from the topology rule.

\begin{lemma}\label{s1prime}
The rule (S1$'$) can be derived:
\begin{equation}\label{eq:S1prime}
\beginpgfgraphicnamed{scalars/s1primeprf1}
\InputIfFileExists{scalars/s1primeprf1.tikz}{}{\input{./figures/scalars/s1primeprf1.tikz}}
\endpgfgraphicnamed=%
\beginpgfgraphicnamed{scalars/s1primeprf5}
\InputIfFileExists{scalars/s1primeprf5.tikz}{}{\input{./figures/scalars/s1primeprf5.tikz}}
\endpgfgraphicnamed
 \end{equation}
\end{lemma}
\begin{proof}
 \begin{align*}
\beginpgfgraphicnamed{scalars/s1primeproof}
\InputIfFileExists{scalars/s1primeproof.tikz}{}{\input{./figures/scalars/s1primeproof.tikz}}
\endpgfgraphicnamed
 \end{align*}
  where we used \eqref{eq:Hopf}, (S1), (S3$'$), the red spider rule, and (IV$'$).
\end{proof}

Now, to derive various rules involving non-trivial phases, we first show some equalities (up to scalars) between red states with phase $\pm\frac{\pi}{2}$ and green states with phases $\mp\frac{\pi}{2}$. These results are very similar to \eqref{eq:Y-state} in Section \ref{s:soundness_simplified}, but \eqref{eq:Y-state} was proved using the rules from Figure \ref{fig:ZX_rules}, whereas we are now using the rule set in Figure \ref{figure3}.
%% check this reference still makes sense after rewrite

\begin{lemma}\label{piby2transformation}
The red state with phase $-\frac{\pi}{2}$ is equal to the green state with phase $\frac{\pi}{2}$ up to some scalars:
 \begin{equation}\label{piby2transform}
\beginpgfgraphicnamed{scalars/piby2trans}
\InputIfFileExists{scalars/piby2trans.tikz}{}{\input{./figures/scalars/piby2trans.tikz}}
\endpgfgraphicnamed
\end{equation}
\end{lemma}
\begin{proof}
We first use \eqref{eq:H_swap} and (EU$'$). Then:
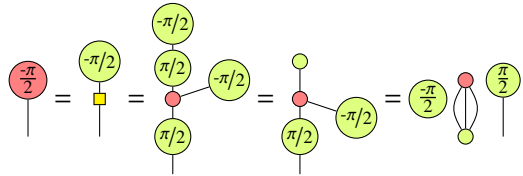
\begin{align*}
\beginpgfgraphicnamed{scalars/piby2transprf1}
\begin{tikzpicture}
	\begin{pgfonlayer}{nodelayer}
		\node [style=none] (0) at (0.5, -0.5) {};
		\node [style=rn] (1) at (0.5, 0.25) {$\frac{\textnormal{-}\pi}{2}$};
	\end{pgfonlayer}
	\begin{pgfonlayer}{edgelayer}
		\draw (1) to (0.center);
	\end{pgfonlayer}
\end{tikzpicture}}
\endpgfgraphicnamed=%
\beginpgfgraphicnamed{scalars/piby2transprf2}
\begin{tikzpicture}
	\begin{pgfonlayer}{nodelayer}
		\node [style=none] (0) at (0.5, -0.5) {};
		\node [style={H box}] (1) at (0.5, 0) {};
		\node [style=gn] (2) at (0.5, 0.5) {$\nicefrac{\textnormal{-}\pi}{2}$};
	\end{pgfonlayer}
	\begin{pgfonlayer}{edgelayer}
		\draw (0.center) to (1);
		\draw (2) to (1);
	\end{pgfonlayer}
\end{tikzpicture}}
\endpgfgraphicnamed=%
\beginpgfgraphicnamed{scalars/piby2transprf3}
\InputIfFileExists{scalars/piby2transprf3.tikz}{}{\input{./figures/scalars/piby2transprf3.tikz}}
\endpgfgraphicnamed=%
\beginpgfgraphicnamed{scalars/piby2transprf4}
\InputIfFileExists{scalars/piby2transprf4.tikz}{}{\input{./figures/scalars/piby2transprf4.tikz}}
\endpgfgraphicnamed=%
\beginpgfgraphicnamed{scalars/piby2transprf5}
\InputIfFileExists{scalars/piby2transprf5.tikz}{}{\input{./figures/scalars/piby2transprf5.tikz}}
\endpgfgraphicnamed
\end{align*}
by (S1), (IV), and the colour-swapped version of (B1).
\end{proof}

Composing with the Hadamard node on both sides of \eqref{piby2transform}, we get:
\begin{equation}\label{piby2transcolour}
\beginpgfgraphicnamed{scalars/piby2transcolorch}
\InputIfFileExists{scalars/piby2transcolorch.tikz}{}{\input{./figures/scalars/piby2transcolorch.tikz}}
\endpgfgraphicnamed
\end{equation}
Furthermore, multiplying both sides of \eqref{piby2transform} by  the red-green scalar and using (IV), we find:
\begin{equation}\label{piby2transinver}
\beginpgfgraphicnamed{scalars/piby2transinverse}
\InputIfFileExists{scalars/piby2transinverse.tikz}{}{\input{./figures/scalars/piby2transinverse.tikz}}
\endpgfgraphicnamed
\end{equation}

\begin{lemma}\label{lem:piby2multip}
Applying a red co-copy map to two green states with phases $-\frac{\pi}{2}$ and $\frac{\pi}{2}$ yields the red state with zero phase:
 \begin{equation}\label{piby2multip}
\beginpgfgraphicnamed{scalars/piby2multiply}
\InputIfFileExists{scalars/piby2multiply.tikz}{}{\input{./figures/scalars/piby2multiply.tikz}}
\endpgfgraphicnamed
\end{equation}
\end{lemma}
\begin{proof}
Using (\ref{piby2transcolour}) and (\ref{piby2transinver}), we have:
\begin{align*}
\beginpgfgraphicnamed{scalars/piby2multiplyprf1}
\InputIfFileExists{scalars/piby2multiplyprf1.tikz}{}{\input{./figures/scalars/piby2multiplyprf1.tikz}}
\endpgfgraphicnamed=%
\beginpgfgraphicnamed{scalars/piby2multiplyprf2}
\InputIfFileExists{scalars/piby2multiplyprf2.tikz}{}{\input{./figures/scalars/piby2multiplyprf2.tikz}}
\endpgfgraphicnamed=%
\beginpgfgraphicnamed{scalars/piby2multiplyprf3}
\InputIfFileExists{scalars/piby2multiplyprf3.tikz}{}{\input{./figures/scalars/piby2multiplyprf3.tikz}}
\endpgfgraphicnamed=%
\beginpgfgraphicnamed{scalars/piby2multiplyprf4}
\begin{tikzpicture}
	\begin{pgfonlayer}{nodelayer}
		\node [style=rn] (0) at (0.5, 0.25) {};
		\node [style=none] (1) at (0.5, -0.5) {};
	\end{pgfonlayer}
	\begin{pgfonlayer}{edgelayer}
		\draw (0) to (1.center);
	\end{pgfonlayer}
\end{tikzpicture}}
\endpgfgraphicnamed
\end{align*}
where the last step is by (S1).
\end{proof}

This is the so-called \emph{supplementary property} for angles $\frac{\pi}{2}$ and $-\frac{\pi}{2}$.

\begin{corollary}\label{cor:piby2scalars}
Two green nodes with no inputs or outputs and phases $-\frac{\pi}{2}$ and $\frac{\pi}{2}$, respectively, are equal to two copies of the red-green scalar:
\begin{equation}\label{piby2saclarmultip}
\beginpgfgraphicnamed{scalars/piby2saclarmultiply}
\InputIfFileExists{scalars/piby2saclarmultiply.tikz}{}{\input{./figures/scalars/piby2saclarmultiply.tikz}}
\endpgfgraphicnamed
\end{equation}
\end{corollary}
\begin{proof}
Using (S1), each green node can be pulled apart into two nodes connected by an edge. Then:
\[
\beginpgfgraphicnamed{scalars/piby2saclarmultiplyprf}
\InputIfFileExists{scalars/piby2saclarmultiplyprf.tikz}{}{\input{./figures/scalars/piby2saclarmultiplyprf.tikz}}
\endpgfgraphicnamed,
\]
by (B1) and \eqref{piby2multip}.
\end{proof}

\begin{corollary}\label{cor:pibyb2swap}
The scalars in \eqref{piby2transform} and \eqref{piby2transcolour} can be brought to the other side:
\begin{equation}
\beginpgfgraphicnamed{scalars/piby2whitetored}
\InputIfFileExists{scalars/piby2whitetored.tikz}{}{\input{./figures/scalars/piby2whitetored.tikz}}
\endpgfgraphicnamed\label{piby2gntored}
\end{equation}
\begin{equation}
\beginpgfgraphicnamed{scalars/piby2redtowhite}
\InputIfFileExists{scalars/piby2redtowhite.tikz}{}{\input{./figures/scalars/piby2redtowhite.tikz}}
\endpgfgraphicnamed\label{piby2redtogn}
\end{equation}
\end{corollary}
\begin{proof}
To prove \eqref{piby2gntored}, start with the right-hand side of that equality, substitute with \eqref{piby2transform} and then use \eqref{piby2saclarmultip}:
\[
\beginpgfgraphicnamed{scalars/piby2whitetoredprf}
\InputIfFileExists{scalars/piby2whitetoredprf.tikz}{}{\input{./figures/scalars/piby2whitetoredprf.tikz}}
\endpgfgraphicnamed
\]
Then \eqref{piby2redtogn} can derived by applying Hadamard gates on both sides of \eqref{piby2gntored} and using (H) and \eqref{eq:H_swap}.
\end{proof}

\begin{lemma}\label{lem:innerprod}
The inner product between a green state of any phase and the red zero-phase effect is equal to the red-green scalar:
 \begin{equation}\label{angledelete}
\beginpgfgraphicnamed{scalars/angledelete}
\InputIfFileExists{scalars/angledelete.tikz}{}{\input{./figures/scalars/angledelete.tikz}}
\endpgfgraphicnamed
\end{equation}
\end{lemma}
\begin{proof}
We prove equality to the red-green scalar for each phase angle in turn. If the phase is $-\frac{\pi}{2}$, substitute for the green state using \eqref{piby2transcolour}. By (H), a red node with no outputs or inputs is equal to a green node with the same phase. The desired result then follows via \eqref{piby2saclarmultip} and (IV). For phase $\frac{\pi}{2}$ we start with \eqref{piby2gntored} and then proceed as before:
\begin{align}
\beginpgfgraphicnamed{scalars/angledeleteprf1}
\InputIfFileExists{scalars/angledeleteprf1.tikz}{}{\input{./figures/scalars/angledeleteprf1.tikz}}
\endpgfgraphicnamed \label{eq:angledelete-pi2}\\
\beginpgfgraphicnamed{scalars/angledeleteprf2}
\InputIfFileExists{scalars/angledeleteprf2.tikz}{}{\input{./figures/scalars/angledeleteprf2.tikz}}
\endpgfgraphicnamed
\end{align}
If the phase is $\pi$, we begin by splitting the green node using (S1) and applying \eqref{piby2transcolour}. We then apply (S1), followed by (B1). The final steps use (H), \eqref{piby2saclarmultip}, (IV), and \eqref{eq:angledelete-pi2}:
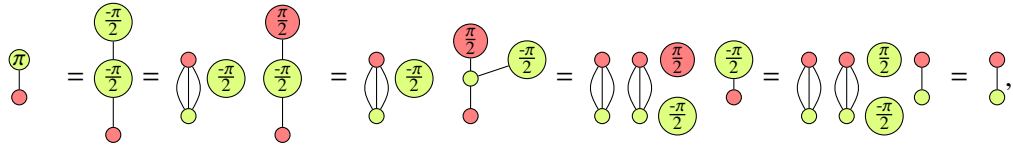
\begin{align*}
\beginpgfgraphicnamed{scalars/angledeleteprf3}
\InputIfFileExists{scalars/angledeleteprf3.tikz}{}{\input{./figures/scalars/angledeleteprf3.tikz}}
\endpgfgraphicnamed,
\end{align*}
This completes the proof.
\end{proof}

\begin{lemma}\label{lem:anglefreepi}
A green $\pi$ phase shift is equal, up to normalisation, to a loop with a Hadamard in it:
 \begin{equation}\label{anglefreepi}
\beginpgfgraphicnamed{scalars/anglefreepi}
\InputIfFileExists{scalars/anglefreepi.tikz}{}{\input{./figures/scalars/anglefreepi.tikz}}
\endpgfgraphicnamed
\end{equation}
\end{lemma}
\begin{proof}
Starting from the right-hand side, use (EU$'$) followed by (S1). Then:
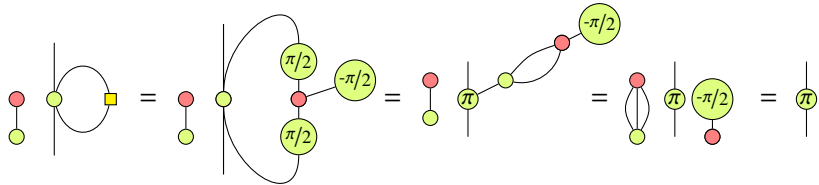
\[
\beginpgfgraphicnamed{scalars/anglefreepiprf}
\InputIfFileExists{scalars/anglefreepiprf.tikz}{}{\input{./figures/scalars/anglefreepiprf.tikz}}
\endpgfgraphicnamed
\]
where the last steps use \eqref{eq:Hopf}, (IV) and \eqref{angledelete}.
\end{proof}

It follows immediately from \eqref{anglefreepi} via the green spider rule and (S3$'$) that a green $\pi$-phase state can also be written in phase-free form:
 \begin{equation}\label{anglefreepidot}
\beginpgfgraphicnamed{scalars/anglefreepidot}
\InputIfFileExists{scalars/anglefreepidot.tikz}{}{\input{./figures/scalars/anglefreepidot.tikz}}
\endpgfgraphicnamed
\end{equation}

\begin{lemma}\label{lem:pidotcopy}
The green $\pi$-phase state is copied by the red copy map, up to normalisation:
 \begin{equation}\label{pidotcopy}
\beginpgfgraphicnamed{scalars/pidotcopy}
\InputIfFileExists{scalars/pidotcopy.tikz}{}{\input{./figures/scalars/pidotcopy.tikz}}
\endpgfgraphicnamed
\end{equation}
\end{lemma}
\begin{proof}
Write the state in the phase-free representation derived in \eqref{anglefreepidot}. The subsequent steps variously use (B2$'$), \eqref{eq:Hadamard_self-inv}, (H) and its colour-swapped version, (S3$'$), as well as (S1) and its colour-swapped equivalent. In the step to the last row we use the Hopf law and in the next step, (H) and \eqref{eq:Hadamard_self-inv}.
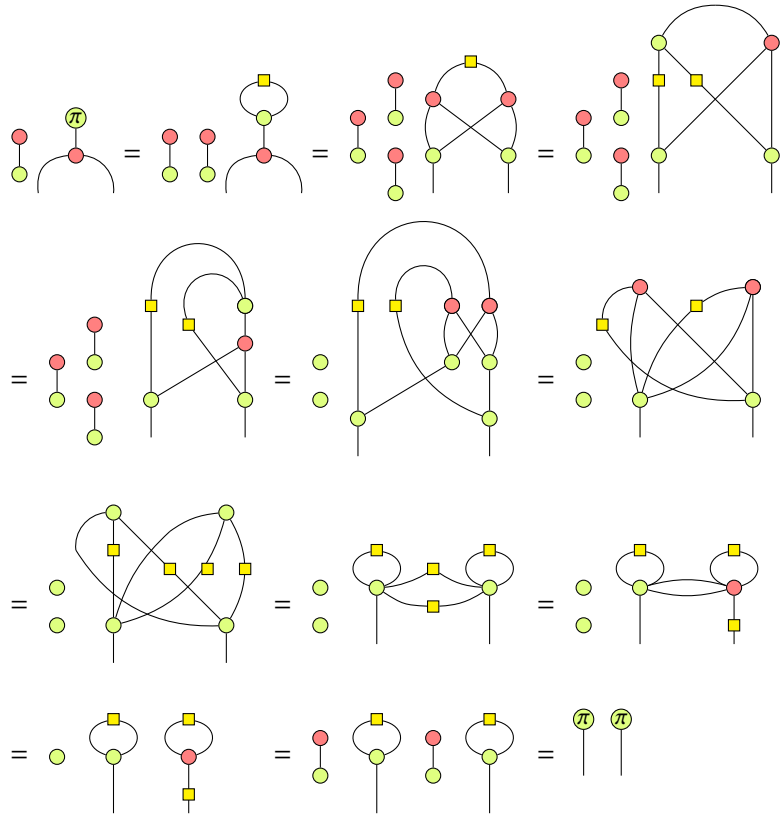
\begin{center}
\beginpgfgraphicnamed{scalars/pidotcopyprf}
\InputIfFileExists{scalars/pidotcopyprf.tikz}{}{\input{./figures/scalars/pidotcopyprf.tikz}}
\endpgfgraphicnamed
\end{center}
Finally, the result follows via \eqref{pidotcopy}.
\end{proof}

\begin{lemma}\label{picopyrul}
(K1) can be derived:
\begin{equation}\label{eq:K1}
\beginpgfgraphicnamed{scalars/k1}
\InputIfFileExists{scalars/k1.tikz}{}{\input{./figures/scalars/k1.tikz}}
\endpgfgraphicnamed
\end{equation}
\end{lemma}
\begin{proof}
By the red spider law, it suffices to prove the following:
\begin{align*}
\beginpgfgraphicnamed{scalars/k1prf1}
\InputIfFileExists{scalars/k1prf1.tikz}{}{\input{./figures/scalars/k1prf1.tikz}}
\endpgfgraphicnamed
\end{align*}
Actually, we have:
\begin{align*}
\beginpgfgraphicnamed{scalars/k1prf2}
\InputIfFileExists{scalars/k1prf2.tikz}{}{\input{./figures/scalars/k1prf2.tikz}}
\endpgfgraphicnamed
 \end{align*}
using \eqref{anglefreepi}, (IV), (S1), (B1), (H), (B2$'$), and \eqref{pidotcopy}.
\end{proof}

\begin{lemma}\label{eulerdecomposition}
(EU) can be derived:
\begin{equation}\label{eq:EU}
\beginpgfgraphicnamed{scalars/HadaDecomSingles}
\InputIfFileExists{scalars/HadaDecomSingles.tikz}{}{\input{./figures/scalars/HadaDecomSingles.tikz}}
\endpgfgraphicnamed
\end{equation}
\end{lemma}
\begin{proof}
Start from the right-hand side. Then:
\begin{align*}
\beginpgfgraphicnamed{scalars/euleroldformprf}
\InputIfFileExists{scalars/euleroldformprf.tikz}{}{\input{./figures/scalars/euleroldformprf.tikz}}
\endpgfgraphicnamed
 \end{align*}
where we used the red spider rule, \eqref{piby2redtogn}, \eqref{piby2transform}, \eqref{piby2saclarmultip} and (IV).
\end{proof}

\begin{lemma}\label{k2rule}
(K2) can be derived:
\begin{equation}\label{eq:K2}
\beginpgfgraphicnamed{scalars/k2s}
\InputIfFileExists{scalars/k2s.tikz}{}{\input{./figures/scalars/k2s.tikz}}
\endpgfgraphicnamed
 \end{equation}
where $\alpha\in \{0, \frac{\pi}{2}, \pi, -\frac{\pi}{2}\}$.
\end{lemma}
\begin{proof}
Using \eqref{angledelete}, it is easy to see that (K2) holds for $\alpha=0$.
For $\alpha=\frac{\pi}{2}$, we have
\begin{align*}
\beginpgfgraphicnamed{scalars/k2ruleprf}
\InputIfFileExists{scalars/k2ruleprf.tikz}{}{\input{./figures/scalars/k2ruleprf.tikz}}
\endpgfgraphicnamed
 \end{align*}
where the first step uses the red spider rule and \eqref{piby2transform}, the second step uses the $\pi$-copy rule \eqref{eq:K1} and the green spider rule, the third step uses \eqref{piby2redtogn} and \eqref{piby2transcolour}, and the last step the red spider rule, \eqref{piby2saclarmultip}, and (IV).

For $\alpha=\pi$, we split the red phase shift and use the above result twice:
\begin{align*}
\beginpgfgraphicnamed{scalars/k2ruleprf2}
\InputIfFileExists{scalars/k2ruleprf2.tikz}{}{\input{./figures/scalars/k2ruleprf2.tikz}}
\endpgfgraphicnamed
\end{align*}
Then, the scalars can be simplified using \eqref{pidotcopy}, the red and green spider rules, and (IV).

For $\alpha=-\frac{\pi}{2}$, we proceed similarly, splitting the red phase shift into $\pi$ and $\frac{\pi}{2}$ and applying both of the above results in sequence.
\begin{align*}
\beginpgfgraphicnamed{scalars/k2ruleprf3}
\InputIfFileExists{scalars/k2ruleprf3.tikz}{}{\input{./figures/scalars/k2ruleprf3.tikz}}
\endpgfgraphicnamed
 \end{align*}
The scalars can then be simplified as before, completing the proof.
\end{proof}

Note that this proves the $\pi$-commutation rule only for phase angles that are integer multiples of $\pi/2$. It is unknown whether or not (K2) is necessary when more general phase angles are allowed.
%% or put this in conclusions only?

\begin{lemma}\label{zeororule}
(ZO) can be derived:
\begin{equation}\label{eq:ZO}
\beginpgfgraphicnamed{scalars/zo1}
\InputIfFileExists{scalars/zo1.tikz}{}{\input{./figures/scalars/zo1.tikz}}
\endpgfgraphicnamed
 \end{equation}
\end{lemma}
\begin{proof}
From the left-hand side of the above, first apply (IV). Then:
\begin{align*}
\beginpgfgraphicnamed{scalars/zeroruleproof}
\InputIfFileExists{scalars/zeroruleproof.tikz}{}{\input{./figures/scalars/zeroruleproof.tikz}}
\endpgfgraphicnamed.
\end{align*}
where the second step uses (ZO$'$), the third step uses the fact that a green node is equal to two copies of the red-green scalar, followed by (IV), (S2$'$) and the red spider rule. The fourth step again uses (ZO$'$). The fifth equality holds by the copy rule. Finally, (ZO$'$) is applied again to complete the proof.
\end{proof}

\end{document}